\documentclass{article}

\usepackage[english]{babel}

\usepackage[letterpaper,top=1in,bottom=1in,left=1in,right=1in,marginparwidth=1.75cm]{geometry}

\usepackage{amsmath,amsthm}
\usepackage{amssymb}
\usepackage{graphicx}
\usepackage[colorlinks=true, allcolors=blue]{hyperref}
\usepackage[capitalize]{cleveref}
\usepackage{multirow}
\usepackage{xspace}
\usepackage{thm-restate}

\newtheorem{theorem}{Theorem}
\newtheorem{lemma}{Lemma}
\newtheorem{observation}{Observation}

\newtheorem*{claim*}{Claim}

\newtheorem{corollary}{Corollary}
\newtheorem{example}{Example}

\theoremstyle{definition}
\newtheorem{definition}{Definition}

\crefname{claim}{Claim}{Claims}
\Crefname{claim}{Claim}{Claims}
\crefname{table}{Table}{Tables}
\Crefname{table}{Table}{Tables}
\crefname{observation}{Observation}{Observations}
\Crefname{observation}{Observation}{Observations}
\crefname{equation}{Equation}{Equations}
\Crefname{equation}{Equation}{Equations}

\newcommand{\probdecision}{$1|T_{\max} \le \ell, \sum U_j \le k|$\xspace}

\usepackage{paralist}
\newcommand{\decprob}[3]{
	\begin{center}
		\begin{minipage}{0.96\linewidth}
			\noindent
			\textsc{#1}
			\begin{compactdesc}
			 \item[\textbf{Input:}]  #2
			 \item[\textbf{Question:}]  #3
			\end{compactdesc}
		\end{minipage}
	\end{center}
}
\newcommand{\optprob}[3]{
	\begin{center}
		\begin{minipage}{0.96\linewidth}
			\noindent
			\textsc{#1}
			\begin{compactdesc}
			 \item[\textbf{Input:}]  #2
			 \item[\textbf{Question:}]  #3
			\end{compactdesc}
		\end{minipage}
	\end{center}
}

\title{Minimizing the Number of Tardy Jobs and Maximal Tardiness on a Single Machine is NP-complete\thanks{Supported by the ISF, grant No.~1070/20.}
}

\author{Klaus Heeger\footnote{Department of Industrial Engineering and Management, Ben-Gurion~University~of~the~Negev, Beer-Sheva, Israel. \texttt{heeger@post.bgu.ac.il}.} \and
Danny Hermelin\footnote{Department of Industrial Engineering and Management, Ben-Gurion~University~of~the~Negev, Beer-Sheva, Israel. \texttt{hermelin@bgu.ac.il}.} \and
Michael L.~Pinedo\footnote{Stern School of Business, New York University, New York City, USA. \texttt{mlp5@stern.nyu.edu}.} \and
Dvir Shabtay\footnote{Department of Industrial Engineering and Management, Ben-Gurion~University~of~the~Negev, Beer-Sheva, Israel. \texttt{dvirs@bgu.ac.il}.}}

\date{}

\begin{document}
\maketitle

\begin{abstract}
This paper resolves a long-standing open question in bicriteria scheduling regarding the complexity of a single machine scheduling problem which combines the number of tardy jobs and the maximal tardiness criteria. We use the lexicographic approach with the maximal tardiness being the primary criterion. Accordingly, the objective is to find, among all solutions minimizing the maximal tardiness, the one which has the minimum number of tardy jobs. The complexity of this problem has been open for over thirty years, and has been known since then to be one of the most challenging open questions in multicriteria scheduling. We resolve this question by proving that the problem is strongly NP-complete. We also prove that the problem is at least weakly NP-complete when we switch roles between the two criteria (i.e., when the number of tardy jobs is the primary criterion). Finally, we provide hardness results for two other approaches (constraint and a priori approaches) to deal with these two criteria.
\end{abstract}

\section{Introduction}
\label{sec:intro}
Since the early stages of classical scheduling theory, the main focus has been the optimization of a single criterion such as the makespan, maximal tardiness, total tardiness, or number of tardy jobs. However, in many practical cases, service and production organizations need to take more than a single objective into account when trying to produce an efficient schedule. For example, a production firm may want to balance its ability to meet job due dates with its ability to control the amount of work-in process held in the shop. Meeting the first objective may prioritize scheduling jobs with an early due date first, while meeting the second objective may prioritize scheduling jobs with short processing times first. As another example, consider a pizzeria that charges no money for late deliveries. Accordingly, the pizzeria owners may try to provide a delivery schedule that minimizes the number of tardy deliveries. However, such a strategy may yield an unfair solution when late deliveries have huge tardiness. In order to produce balanced delivery schedules, they may consider the maximal tardiness as an additional criterion to evaluate the quality of a delivery schedule. Given the above deficiency of traditional scheduling models, the field of multicriteria scheduling has gained a lot of attention from the late 80's on. Ever since, the literature on multicriteria scheduling has expanded considerably, with several survey papers and books published over the years~\cite{MultiagentSchedulingBook,DBLP:journals/eor/Hoogeveen05,LeeV93,DBLP:journals/rairo/TKindtB01,MulticriteriaSchedulingBook}. 

Consider two different minimization objectives $\gamma_1$ and $\gamma_2$ for a given scheduling problem. In a typical setting there are not enough resources to compute the entire set of Pareto optimal solutions, as this set is usually quite large. Therefore, one needs a way to define which solution is the most desired in this set. There are essentially three established approaches to tackle this issue, each of which defines a different problem for a given pair of scheduling criteria $\gamma_1$ and $\gamma_2$ (see, e.g.,~\cite{DBLP:journals/eor/Hoogeveen05}).
\begin{itemize}
\item The \emph{lexicographic} approach: Find a solution that minimizes $\gamma_2$ among all solutions that minimize~$\gamma_1$. This variant is usually denoted by~$1||Lex(\gamma_1,\gamma_2)$ in the single machine setting, where $\gamma_1$ is called the \emph{primary criterion} and $\gamma_2$ is the \emph{secondary criterion}.
\item  The \emph{constraint} approach: Given a threshold parameter~$\ell$, find a solution that minimizes $\gamma_1$ subject to the constraint that $\gamma_2 \le \ell$. In the single machine setting with no additional constraints this variant is usually denoted by~$1|\gamma_2 \le \ell |\gamma_1$.
\item The \emph{a priori} approach: Find a solution that minimizes $\alpha \gamma_1+\gamma_2$, where $\alpha$ is a given constant that indicates the relative importance of criterion~$\gamma_1$ with respect to criterion $\gamma_2$. This variant is usually denoted by $1||\alpha \gamma_1+\gamma_2$ in the single machine setting. 
\end{itemize}
Note that the lexicographic approach is a special case of the a priori approach. 

In this paper we consider two of the most basic and classical scheduling objectives: The first is the maximal tardiness criterion, typically denoted by~$T_{\max}$, which measures the maximum tardiness of any job in the schedule. It is well-known that the $1||T_{\max}$ problem, the problem of minimizing $T_{\max}$ on a single machine, is solvable in $O(n\log n)$ time by processing the jobs based on the Earliest Due Date (EDD) rule~\cite{Jackson1956}. The second criterion, denoted $\sum U_j$, is the total number of tardy jobs in the schedule. The corresponding single machine $1||\sum U_j$ problem is also solvable in $O(n\log n)$ time, due to an algorithm presented by Moore in the late 60s~\cite{Moore1968}. 
It follows that minimizing either the maximal tardiness or the number of tardy jobs can be done in polynomial time when the scheduling is done on a single machine. However, when both criteria are considered together, the resulting bicriteria problems become much harder to analyze and solve. In fact, the computational complexity status of the problem using either the constraint, lexicographic, or a priori approach was mentioned as an open problem by several different authors. 

Lee and Vairaktarakis~\cite{LeeV93} published in 1993 an influential survey on bicriteria single machine scheduling problems. They focused on the lexicographic approach, and mentioned several open problems involving either one or both of the $T_{\max}$ and $\sum U_j$ criteria. Later on, all of these were resolved by Huo \emph{et al.}~\cite{Huo2,HuoLZ07}, apart from $1||Lex(T_{\max},\sum U_j)$ and $1||Lex(\sum U_j,T_{\max})$ which were left open. This is summarized in~\cite{HuoLZ07} with the following quote:
\begin{quote}
\begin{center}
\textit{``Despite much efforts spent on $1||Lex(T_{\max}, \sum U_j) $ and $1||Lex(\sum U_j, T_{\max})$, their complexity remain open. Although we cannot prove it, we conjecture that they are both NP-hard. It will be worthwhile to settle this issue in the future."}  
\end{center}    
\end{quote}
The complexity status of the problem was mentioned as open also in later surveys by T’kindt and Billaut~\cite{DBLP:journals/rairo/TKindtB01} and Hoogeveen~\cite{DBLP:journals/eor/Hoogeveen05}, and also in the book on multicriteria scheduling by T’kindt and Billaut~\cite{MulticriteriaSchedulingBook}. The $1|T_{\max} \leq \ell |\sum U_j$ problem is mentioned in the book on multiagent scheduling by Agnetis \emph{et al.}~\cite{MultiagentSchedulingBook}, who wrote 
\begin{quote}
\begin{center}
\textit{``The complexity of this problem still stands out as one of the most prominent open issues in theoretical scheduling."}    
 \end{center}
\end{quote}
As such, the complexity status of these problems have been established as the main open problem in multicriteria scheduling.

\subsection{Our Results}

We determine the computational complexity of single machine bicriteria scheduling involving objectives~$T_{\max}$ and $\sum U_j$ using either the constraint, lexicographic, or a priori approach, by showing that all problems are unlikely to admit polynomial-time algorithms. 

Our first main result involves the constraint variant of the problem. In its decision form, the problem asks to determine whether there exists a schedule with $T_{\max} \leq \ell$ and $\sum U_j \leq k$, and we denote this problem by~$1|T_{\max} \le \ell, \sum U_j \le k|$. We prove that this problem is strongly NP-complete by a reduction from \textsc{3-Partition}. 
\begin{restatable}{theorem}{ThmStrong}
\label{thm:strong}%
$1|T_{\max} \le \ell, \sum U_j \le k|$ is strongly NP-complete.
\end{restatable}
\noindent Following this, we show that there exists an easy reduction from $1|T_{\max} \le \ell, \sum U_j \le k|$ to the $1||Lex(T_{\max}, \sum U_j)$ problem, as well an easy reduction from $1||Lex(T_{\max}, \sum U_j)$ to the $1||\alpha T_{\max}+\sum U_j$ problem. As both reductions preserve strong NP-hardness, this directly yields: 
\begin{corollary}
$1||Lex(T_{\max}, \sum U_j)$ and $1||\alpha T_{\max}+\sum U_j$ are both strongly NP-complete.
\end{corollary}

Unfortunately, we could not find a direct reduction from $1|T_{\max} \le \ell, \sum U_j \le k|$ to $1||Lex (\sum U_j, T_{\max})$. Thus, we are forced to design an alternative reduction for this last variant. While we cannot find such a reduction from a strong NP-hard problem, we are still able to devise a reduction from the weakly NP-complete \textsc{Partition} problem, giving us our second main result of the paper. 
\begin{theorem}
\label{thm:weak}%
$1||Lex(\sum U_j, T_{\max})$ is weakly NP-complete.
\end{theorem}
\noindent Thus, altogether our results resolve the complexity of all single machine bicriteria problems involving the $T_{\max}$ and $\sum U_j$ criteria, answering the long standing open question posed in~\cite{MultiagentSchedulingBook,DBLP:journals/eor/Hoogeveen05,HuoLZ07,Huo2,LeeV93,DBLP:journals/rairo/TKindtB01,MulticriteriaSchedulingBook}.

\subsection{Related Work}

Shantikumar~\cite{Shanthikumar} designed a branch-and-bound procedure for $1||Lex(\sum U_j, T_{\max})$ which requires exponential time in the worst case. He also presented a polynomial-time algorithm for minimizing the maximal tardiness when the set of early jobs is given in advance. Chen and Bulfin~\cite{ChenBul} designed a branch-and-bound procedure for solving $1||Lex(T_{\max}, \sum U_j)$. By applying a set of numerical tests, they showed that their algorithm was able to solve instances of up to 40 jobs in less than one minute of computer time.
Finally, Huo \emph{et al.}~\cite{HuoLZ07} implemented several heuristics for $1||Lex (T_{\max}, \sum U_j)$ and tested them on instances with up to 200 jobs; their best heuristic was on average less than 1\% worse than the optimal.

Huo \emph{et al.}~\cite{HuoLZ07} also showed that the weighted problems $1|| Lex(\max w_jT_j, \sum U_j)$ and $1|| Lex (\sum U_j, \max w_jT_j)$, where each job has its own weight and one of the two criteria is to minimize the maximal weighted tardiness, are both weakly NP-complete. The results in this paper supersede both of these hardness results, as they are for the unweighted versions of these problems. The constraint problem $1|T_{\max} \le \ell | \sum U_j$ is a special case of $1 | \overline{d}_j|\sum U_j$. In this problem, each job has a due date~$d_j$ as well as a deadline~$\overline{d}_j$ which must be met. The goal is to minimize the number of tardy jobs while meeting all deadlines. When $\overline{d}_j = d_j + \ell$ for each job~$j$, this problem becomes $1|T_{\max} \le \ell | \sum U_j$. Yuan showed that~$1 | \overline{d}_j| \sum U_j$
is strongly NP-complete~\cite{Yuan17}, and this result is superseded by our results as well.

\section{Preliminaries}
\label{sec:prelims}

We use standard terminology from computer science and scheduling.
For each $n \in \mathbb{N}$, we denote by $[n]:= \{1, \ldots, n\}$ the set of positive integers up to (and including) $n$.
For a bijective function~$f:X \rightarrow Y$, we denote its inverse by $f^{-1}$; thus, we have $f^{-1} (y) =x$ if and only if $f(x) = y$.

\paragraph*{Scheduling.}
We consider scheduling problems involving a set of $n$ jobs $\mathcal{J}$; where all are available at time zero to be non-preemptively processed on a single machine. Each job~$J \in \mathcal{J}$ has two nonnegative, integral parameters: its \emph{processing time} $p (J)$ and its \emph{due date} $d (J)$.
A \emph{schedule} is a bijection from $[n]$ to the set~$\mathcal{J}$ of jobs, where the job~$J\in \mathcal J$ with $\sigma (i) = J$ is the $i$-th job to be processed on the single machine.

Given a schedule~$\sigma$, the \emph{completion time} of job~$J$ is $C_{\sigma} (J) := \sum_{i=1}^{\sigma^{-1} (J)} p (\sigma (i))$.
We will focus on the following two objectives common in the scheduling literature: $T_\sigma (J) :=\max\{0,C_\sigma (J)-d(J)\}$, which is the \emph{tardiness} of job $J$; and $U_\sigma (J)$ which is a binary tardiness indicator that equals 1 if $C_{\sigma} (J) > d (J)$ and 0 otherwise.
If the schedule $\sigma $ is clear from context, then we may also drop the subscript and just write $C (J)$, $T(J)$, or $U(J)$.
If $U(J) =1$, then we say that job $J$ is \emph{tardy}. Otherwise, $J$ is an \emph{early} job.

We study the scheduling problems arising from combining the objective total number of tardy jobs~($\sum U_j$) with the maximum tardiness~$T_{\max} = \max_{J \in \mathcal J} T (J)$ using the lexicographic, constraint, or a priori approach as described in \Cref{sec:intro}.
We will examplarily describe the problem resulting from the constrained approach using $T_{\max} $ as bounded objective:

\optprob{$1|T_{\max} \le \ell| \sum U_j$}{
A set~$\mathcal J$ of jobs with processing times $p : \mathcal J\rightarrow \mathbb{N}$, due dates $d: \mathcal J \rightarrow \mathbb{N}$, and an integer~$\ell \in \mathbb{N}$.
}{
Find a schedule~$\sigma$ minimizing $\sum_{J \in \mathcal J} U (J)$ among all schedules with $\max_{J \in \mathcal J} T (J) \le \ell$.
}
For the constrained version, we will call a schedule \emph{feasible} if it obeys the constraint.
For example, in the context of problem $1|T_{\max} \le \ell| \sum U_j$, a schedule is feasible if $T (J) \le \ell$ for every job~$J \in \mathcal{J}$.
A schedule is \emph{optimal} if it is feasible and, among all feasible schedules, minimizes the objective.
Taking again $1|T_{\max} \le \ell| \sum U_j$ as an example, a schedule is optimal if $T (J) \le \ell$ for every job $J \in \mathcal J $ and among all such schedules, it minimizes the number of tardy jobs.
Note that the lexicographic version is a special case of the constrained version where the upper bound on the constrained objective is the minimum value of this objective over all schedules;
thus, the notions of feasible and optimal schedule extend to the lexicographic case, where a schedule is feasible if it minimizes the primary objective.

\paragraph*{Canonical schedules.}

In order to minimize the maximum tardiness, one schedules the jobs according to non-decreasing due date~\cite{Jackson1956}.
By using modified due dates, this observation can easily be extended to finding a schedule with maximum tardiness~$\ell$ subject to the condition that a given subset~$\widetilde{\mathcal{J}}$ of jobs is early (see e.g.\ \cite[Observation 3.2]{HeegerHS23}). This leads to following definition which is central to our paper:
\begin{definition}[\textbf{Canonical schedule}]
\label{def:canonical}
Let $\mathcal{\widetilde J}$ be a subset of jobs. The \emph{canonical schedule} for $\mathcal{\widetilde J}$ is the schedule where jobs are processed in non-decreasing order of the modified due date~$\widetilde{d}$ defined by  
\[
\widetilde{d}(J) := 
\begin{cases}
d(J) & \text{if } J \in\widetilde{\mathcal J}\\
d(J) + \ell & \text{if } J \notin\widetilde{\mathcal J}
\end{cases}
\]
\end{definition}
\noindent The name canonical is justified due to the fact that if there exists a schedule with a maximum tardiness at most $\ell$ where each job from~$\widetilde{\mathcal J}$ is early, than there also such a schedule which is canonical. 
\begin{observation}\cite[Observation 3.2]{HeegerHS23}\label{obs:almost-edd}\label{lem:generalized-EDD}
    Let $\mathcal{ \widetilde J}$ be a subset of jobs.
        If there is a schedule with maximum tardiness at most $\ell$ where each job from~$\widetilde{\mathcal{J}}$ is early, then the canonical schedule for $\widetilde{\mathcal{J}}$ is such a schedule.
\end{observation}
\noindent Note that given a subset $\mathcal{\widetilde J}$ of jobs, a schedule minimizing $T_{\max}$ among all schedules where each job from~$\widetilde{\mathcal J}$ is early can be found using \Cref{obs:almost-edd} and binary search.

\section{Strong NP-completeness}
\label{sec:strong-NP}
\newcommand{\blocklength}{\Delta}
\newcommand{\firstblocklength}{\blocklength_1}
\newcommand{\secondblocklength}{\blocklength_2}
\newcommand{\partitiontarget}{t}
\newcommand{\firsthalfend}{D^*}

This section describes a proof for \Cref{thm:strong}:
\ThmStrong*
\noindent In order to show \Cref{thm:strong}, we reduce from \textsc{3-Partition} which is strongly NP-complete~\cite{GareyJohnson}.
\decprob{3-Partition}{
A multiset of $n$ integers~$a_1, \ldots, a_{n}$.
}{
Is there a partition~$(S_1, \ldots, S_m)$ of $[n]$ where $n = 3m$ such that $\sum_{i \in S_j} a_i = \partitiontarget$ for every $j \in [m]$ where $\partitiontarget:= \sum_{i=1}^{n} a_i /m$?
}
\noindent 
We will interpret $1|\sum U_j \le k, T_{\max} \le \ell|$ as the decision version of the constrained $1|T_{\max} \le \ell|\sum U_j$ problem, implying that the feasibility of a schedule refers to its maximum tardiness, i.e., we call a schedule feasible if its maximum tardiness is at most $\ell$. We will use throughout our construction a sufficiently large constant~$\alpha$ defined by $\alpha = 10n^2 \cdot t$.

For ease of presentation, we partition the timeline into $m$ \emph{time periods}, each corresponding to a different index $j \in [m]$. Define $\delta$ to be the value
$$
\delta = 4n\cdot  \alpha^3 + 2 \cdot \alpha^2 +  (2m + 1)t \cdot \alpha + mt \,.
$$
The total length of each time period will be $\delta$, where the first period starts at time $\Delta_0=0$. In this way, the $j$-th time period starts at time $\Delta_{j-1}$ and ends at time $\Delta_j = \Delta_{j-1} + \delta$. Each period will consist of two \emph{halves}.
The length of the first half of the $j$-th time period will equal
$$
\delta^*_j = 2n \cdot \alpha^3 + \alpha^2  +  (2m -2j + 1)t   \cdot \alpha + (m - j)t\,. 
$$
In this way, the first half of the $j$-th time period ends at time $\Delta^*_j=\Delta_{j-1}+\delta^*_j$, which is the time when the second half starts.
Note that the length of the first and second half of each period is not exactly but only roughly half of the length of the whole period (more specifically, each half has length half of the whole period if ignoring the lower-order terms multiplied by $\alpha$ or $\alpha^0$). 

\subsection{Construction}
\label{sec:construction}

We next describe how to construct an instance of \probdecision from a given instance of \textsc{3-Partition}. The high level idea is as follows: We will construct $m$ \emph{groups} of jobs, one for each index~$j \in [m]$, where the $j$-th group will encode all the integers in $a_1,\ldots,a_n$ that are selected for the solution sets $S_1 \cup \ldots \cup S_j$. Let $j \in [m]$. The $j$-th job group will consist of the $4n$ jobs $J_{i,j}$, $\neg J_{i,j}$, $J^*_{i,j}$, and $\neg J^*_{i,j}$ where $i \in[n]$. Intuitively speaking, scheduling~$J_{i,j}$ and $J^*_{i,j}$ early will correspond to the situation where $a_i \in S_1 \cup \ldots \cup S_j$, while scheduling~$\neg J_{i,j}$ and~$\neg J^*_{i,j}$ early corresponds to the situation where $a_i \notin S_1 \cup \ldots \cup S_j$. Further, we add so-called \emph{delimiter jobs}~$D_j^*$ and $D_j$, ensuring that $\sum_{i \in S_1\cup \ldots \cup S_j} a_i \ge j \cdot \partitiontarget$ respectively $\sum_{i \in S_1 \cup  \ldots\cup S_j} a_i \le j \cdot \partitiontarget$ (see \Cref{lem:early-delimiter-jobs}). Together, accounting for all $j \in [m]$, this will ensure that the sum of integers in each~$S_j$ is exactly $t$.

\paragraph{Job construction.} For each $i \in [n]$ and each $j \in [m]$, we construct four \emph{number jobs} $J_{i,j}$, $\neg J_{i,j}$, $J^*_{i,j}$, and~$\neg J^*_{i,j}$, with the following processing times and due dates:
\begin{itemize}
\item $p(J^*_{i,j})=\alpha^3$ and $d(J^*_{i,j})=\Delta_{j-1} +  2i \cdot \alpha^3+ 0.1 \cdot \alpha^2 $.
\item $p(\neg J^*_{i,j})=\alpha^3 +  a_i$ and $d(\neg J_{i,j}^*)=\Delta_{j-1} + (2i-1)\cdot \alpha^3+ 0.1 \cdot \alpha^2$.
\item $p(J_{i,j}) = \alpha^3 + a_i \cdot \alpha$ and $d(J_{i,j}) = \Delta^*_j + 2i \cdot \alpha^3 + 0.1 \cdot \alpha^2 $.
\item $p(\neg J_{i,j}) = \alpha^3 $ and $d(\neg J_{i,j}) = \Delta^*_j + (2i-1) \cdot \alpha^3+ 0.1 \cdot \alpha^2 $.
\end{itemize}
Note that we have \begin{align*}
d(J_{n, j-1}) &< d(\neg J_{1,j}^*) < d(J_{1,j}^*) < d(\neg J_{2,j}^*)< d(J_{2,j}^*) < \ldots < d( \neg J_{n, j}^*) < d(J_{n, j}^*) \\
& < d(\neg J_{1,j}) < d(J_{1,j}) < d(\neg J_{2,j}) < d(J_{2,j}) < \ldots < d(\neg J_{n, j}) < d(J_{n, j}) < d(\neg J_{1, j+1}^*).
\end{align*}

Furthermore, for each $j \in [m]$, we add two \emph{delimiter jobs} $D_j^*$ and $D_j $ in addition to the $4n$ jobs defined above. These jobs are used to indicate the end of the first and second halves of the $j$-th period. That is, time~$\Delta^*_j$ and time~$\Delta_j$, respectively. The characteristics of $D_j^*$ and $D_j$ are defined as follows:
\begin{itemize}
\item $p(D_j^*) = \alpha^2 + (m - j) t  \cdot  \alpha$ and $d(D_j^*) = \Delta^*_j $, and
\item $p(D_j) = \alpha^2 + j t \cdot \alpha $ and $d(D_j) = \Delta_j$.
\end{itemize}

We will also need to add some filler jobs for the first and last time period. We start by constructing the job~$F_0$ with processing time and due date $p(F_0)=d(F_0)= m \cdot t \cdot \alpha$. The role of this job is to offset all due dates by $m \cdot t \cdot \alpha$, resulting in these due dates being expressible by slightly simpler formulas.
Thus, this job will always be scheduled first in any feasible schedule for the entire \probdecision instance. Furthermore, we construct $n$ filler jobs $F^1_1,\ldots,F^1_n$ for the first time period, and $n$ filler jobs $F^m_1,\ldots,F^m_n$ for the last time period. These filler jobs will have the following characteristics:
\begin{itemize}
\item $p(F_i^1) = \alpha^3$ and due date $d(F_i^1) = d(J^*_{i, 1})$ for each $i \in [n]$ (so $F_i^1$ is a copy of $J^*_{i,1}$), and 
\item $p(F_i^m) = \alpha^3$ and $d(F_i^m) = \Delta_m + 2i \cdot \alpha^3 + 0.1 \cdot \alpha^2$ for each $i \in [n]$.
\end{itemize}

\begin{table}[h!]
\begin{center}
\begin{tabular}{c|c|c}
Job & Processing Time & Due Date\\
\hline
\hline
$J^*_{i,j}$ & $\alpha^3$ & $\Delta_{j-1} + 2i \cdot \alpha^3  + 0.1 \cdot \alpha^2 $\\
$\neg J^*_{i,j}$ & $\alpha^3 + a_i$ & \quad $\Delta_{j-1} + (2i-1)\cdot \alpha^3 + 0.1 \cdot \alpha^2 $ \quad \\
\hline
$D_j^*$ & $ \quad \alpha^2 + (m - j)t \cdot \alpha$ \quad & $\Delta^*_j$\\
\hline
\hline
$J_{i,j}$ & $\alpha^3 +  a_i \cdot \alpha$ & $\Delta^*_j+ 2i \cdot \alpha^3 + 0.1 \cdot \alpha^2 $\\
$\neg J_{i,j}$ & $\alpha^3 $ & $\Delta^*_j + (2i-1) \cdot \alpha^3+ 0.1 \cdot \alpha^2  $\\
\hline
$D_j$ & $ \alpha^2 + jt \cdot \alpha$ & $\Delta_j$\\
\hline
\hline
$F_0$ & $ m t  \cdot \alpha$ & $m t \cdot \alpha$\\
\hline
$F_i^1$ & $\alpha^3$ & $d(J^*_{i,1})$\\
\hline
$F_i^m$ & $\alpha^3$ & $\Delta_m + 2i \cdot \alpha^3 + 0.1 \cdot \alpha^2$
\end{tabular}
\end{center}
\caption{The processing times and due dates of all jobs constructed for the \probdecision instance.
}
\label{tab:jobs}
\end{table}

\paragraph{Parameters $\boldmath{k}$ and $\boldmath{\ell}$.}

We have described above all the jobs constructed for the \probdecision instance. For an overview of their processing times and due dates, see \Cref{tab:jobs}. To finish our construction, we set the target number~$k$ of tardy jobs to
$$
k \,\,=\,\, 2 m \cdot n,
$$
and the target maximum tardiness $\ell$ to
$$
\ell \,\,=\,\, 2n \cdot \alpha^3 + \alpha^2 + 0.1 \cdot \alpha^2 .
$$
Intuitively, the target maximum tardiness allows each job to be tardy by approximately half a period.

\paragraph{Proof roadmap.}

We now briefly sketch the structure of the proof. The main technical difficulty is to show that it indeed suffices to consider canonical schedules for job sets corresponding to an encoding of a solution to the \textsc{3-Partition} instance as described above. That is, job sets containing either both of $J_{i,j}$ and $J_{i, j}^*$ or both of $\neg J_{i,j}$ and $\neg J_{i,j}^*$ for every $i\in [n]$ and $j \in [m]$, and whenever $J_{i,j}$ and $J_{i,j}^*$ are contained in the set, then also $J_{i ,j+1}$ and $J_{i, j+1}^*$ are contained in the set.
This will be proven in \Cref{sec:candidates} and \Cref{lem:only-candidate}, and the corresponding sets of jobs are called \emph{candidate sets}.
Before formally describing candidate sets, we give necessary definitions and helpful observations in \Cref{sec:jobsets}. Assuming that it is indeed sufficient to consider the canonical schedules for candidate sets, we will show the correctness of the reduction in \Cref{sec:correctness}.
We finally prove that we may restrict ourselves to canonical schedules for candidate sets in \Cref{sec:lem5proof}.

\subsection{Job sets $\boldmath{\mathcal{J}_{\le j}}$ and $\boldmath{\mathcal{J}^*_{\le j}}$}
\label{sec:jobsets}

We now introduce notation and observations useful for the proof of correctness.
In particular, we will make frequent use of the following sets of jobs, containing the non-filler and non-delimiter jobs with due dates in the first respectively second half of the $j$-th period: 
For each $j \in [m]$, we use $\mathcal{J}^*_j := \{J_{i,j}^*, \neg J_{i,j}^* : i \in [n]\}$ and $\mathcal{J}_j := \{J_{i,j}, \neg J_{i,j} : i \in [n]\}$.
Next, we define
$\mathcal{J}_{\le 0}=\{F_0\} \cup \{F_i^1 : i \in [n]\}$. Then, for each $j \in [m]$, we define
$$
\mathcal{J}^*_{\le j} := \mathcal{J}_{\le j-1} \cup \mathcal{J}^*_{j} \cup \{  D^*_{j}\}
$$
and
$$
\mathcal{J}_{\le j} := \mathcal{J}_{\le j}^* \cup \mathcal{J}_{j} \cup \{  D_{j}\}. 
$$
In other words, for $j > 0$, $\mathcal{J}_{\le j}$ contains all jobs with due dates at most $\Delta_j$, whereas $\mathcal{J}_{\le j}^*$ contains all jobs with due dates at most $\Delta_j^*$.

We first observe that in any feasible schedule, all jobs from $\mathcal J_{\le j}^*$ except $D_j^*$ must be completed by the end of the $j$-th period (i.e., by time $\Delta_j$), and all jobs from $\mathcal J_{\le j}$ except for $D_j$ must be completed by the end of the first half of the period~$j+1$ (i.e., by time~$\Delta_{j + 1}^*$).
\begin{observation}
\label{eq:negJ}%
 For each $j \in [m]$, we have 
\[\max_{J \in \mathcal J_{\le j}^* \setminus \{D_j^*\}} d(J) + \ell < \Delta_j.\]
\end{observation}
\begin{proof}
The observation follows by easy calculations:
\begin{align}
\max_{J \in \mathcal J_{\le j}^* \setminus \{D_j^*\}} d(J) + \ell & = d(J^*_{n, j}) + \ell \notag\\
& = \Delta_{j -1} + 2 n \cdot \alpha^3 + 0.1 \cdot \alpha^2 + 2n\cdot \alpha^3+ 1.1 \cdot \alpha^2 \notag \\
& = \Delta_{j -1} + 4n\cdot \alpha^3  + 1.2 \cdot \alpha^2  \notag\\
& = \Delta_j - \delta + 4n\cdot \alpha^3  + 1.2 \cdot \alpha^2  \notag\\
& = \Delta_j - 0.8 \cdot \alpha^2 - (2m +1)  t\cdot \alpha  - m t< \Delta_j. \qedhere 
\end{align}
\end{proof}
\begin{observation}
    \label{eq:negtildeJ}
     For each $j \in [m]$, we have 
    \[\max_{J \in \mathcal J_{\le j} \setminus \{D_j\}} d(J) + \ell < \Delta_{j+1}^*.\]
\end{observation}
\begin{proof}
    The proof follows by calculations analogous to the proof of \Cref{eq:negJ}:
    \begin{align}
\max_{J \in \mathcal J_{\le j} \setminus \{D_j\}} d(J) + \ell & = d(J_{n, j}) + \ell\notag \\
& = \Delta_{j}^* + 2n \cdot \alpha^3 + 0.1 \cdot \alpha^2 + 2n\cdot \alpha^3  + 1.1 \cdot \alpha^2 \notag\\
& = \Delta_{j}^* + 4n\cdot \alpha^3 +1.2 \cdot \alpha^2 \notag\\
& = \Delta_{j+ 1}^* - \delta + mt + 4n\cdot \alpha^3 +1.2 \cdot \alpha^2 \notag\\
& = \Delta_{j+1}^* - 0.8 \cdot \alpha^2  - (2m - 1) t \cdot \alpha - (m-1) t < \Delta_{j+1}^*.\notag\qedhere
\end{align}
\end{proof}

We next compare the total processing time of all jobs in $\mathcal{J}_{\le j}^*$ to $\Delta_j$. This value represents, assuming that all jobs from $\mathcal{J}_{\le j}^*$ are completed by the end of the $j$-th period (which is true in all feasible schedules for $\mathcal{J}_{\le j}^* \setminus \{D_j^*\}$ by \Cref{eq:negJ}), how much processing time is left for other jobs.
Similarly, we compare the total processing time of all jobs in $\mathcal{J}_{\le j}^*$ to $\Delta_j^*$, representing the time by which all jobs from $\mathcal{J}_{\le j}^*$ exceed the first $j-1$ periods and the first half of the $j$-th period.
The resulting identity will be used frequently:
\begin{observation}\label{obs:deadline}\label{eq:tildeJ_Delta*}
    For each $j \in [m]$, we have 
    \begin{align*}
        p(\mathcal{J}^*_{\le j}) &= \Delta_j - (n \cdot \alpha^3 + \alpha^2 + 2j t \cdot \alpha)\\
        & =\Delta_j^* + n \cdot \alpha^3 + jt \,.
    \end{align*}
\end{observation}
\begin{proof}
The proof is by induction on~$j$. For $j = 1$, we have 
\begin{align*}
p(\mathcal{J}_{\le 1}^* ) & = \sum_{i= 1}^{n} \bigl(p(J^*_{i, 1}) + p(\neg J^*_{i, 1})\bigr)+ p(D^*_{1}) + \sum_{i =1}^n p(F_i^1) + p(F_0)\\
& = 2n \cdot \alpha^3 + m t + \alpha^2 +  (m-1) t \cdot \alpha + n \cdot \alpha^3 + m t \cdot \alpha\\
& = 3n \cdot \alpha^3 + \alpha^2 + (2m - 1) t \cdot \alpha + m \cdot t  = \Delta_1 - n \cdot \alpha^3 - \alpha^2 - 2t\cdot \alpha\,.
\end{align*}

For $j> 1$, note that
\[
\mathcal{J}_{\le j}^* \setminus \mathcal{J}^*_{\le j-1} = \mathcal{J}^*_j \cup \mathcal J_{j-1} \cup \{D_j^*, D_{j-1}\}\,.
\]
Thus, we have (using induction for the second equality)
\begin{align*}
p(\mathcal{J}^*_{\le j}) & = p(\mathcal{J}^*_{\le j-1}) + p (\mathcal{J}^*_{\le j} \setminus \mathcal{J}^*_{\le j-1}) \\
& = \Delta_{j-1} - \Bigr( n\cdot \alpha^3 + \alpha^2 + 2(j-1) t\cdot \alpha\Bigr) + p(\mathcal{J}_j^*) + p(\mathcal{J}_{j-1}) + p(D_j^*) + p(D_{j-1})\\
& = \Delta_{j-1} - n \cdot \alpha^3  - \alpha^2 - 2(j-1) t\cdot \alpha + 2n \cdot \alpha^3 + m t + 2n \cdot \alpha^3 + m  t\cdot \alpha \\
& \qquad + \alpha^2+ (m-j)t\cdot \alpha + \alpha^2 +  (j-1) t \cdot \alpha\\
& = \Delta_{j-1} + 3n\cdot \alpha^3  + \alpha^2 + (2m -2j + 1)t \cdot \alpha +  m t \\
& = \Delta_j - n \cdot \alpha^3 - \alpha^2 -  2j t\cdot \alpha \,. 
\end{align*}

For the second equation, we have
    \begin{align*}
        p(\mathcal{J}^*_{\le j}) & = \Delta_j - (n \cdot \alpha^3 + \alpha^2 + 2jt \cdot \alpha)\notag\\
        & = \Delta_j^* + 2n \cdot \alpha^3 + \alpha^2 + 2jt \cdot \alpha + jt -n \cdot \alpha^3 - \alpha^2 - 2jt \cdot \alpha\notag\\
        & = \Delta_j^* + n \cdot \alpha^3 + jt\,.\qedhere
    \end{align*}
\end{proof}

Analogously to \Cref{obs:deadline}, we now compare the total processing time of all jobs from $\mathcal{J}_{\le j-1})$ to the end of the first half of the $j$-th period and the end of the $(j-1)$-th period:

\begin{observation}
\label{obs:ptildeJ}\label{eq:Jj-1_Dj-1}%
For each $j \in [m]$, we have
\begin{align*}
    p(\mathcal{J}_{\le j-1}) &= \Delta^*_j -( n\cdot \alpha^3  + \alpha^2 + (m - j) t\cdot \alpha + (m - j)  t)\\
    & =\Delta_{j-1} + n \cdot \alpha^3 + (m-j+1)t \cdot \alpha \,.
\end{align*}
\end{observation}

\begin{proof}
We have (using \Cref{obs:deadline} for the second equality)
\begin{align*}
p(\mathcal{J}_{\le j-1}) & = p(\mathcal{J}_{\le j}^*) - \Bigl(p(\mathcal{J}_j^*) + p(D_j^*)\Bigr)\\
& = \Delta_j - \Bigl(n \cdot \alpha^3 + \alpha^2 + 2jt\cdot \alpha\Bigr) - \Bigl(2n\cdot \alpha^3 + m t + \alpha^2 + (m -j)t\cdot \alpha\Bigr)\\
& = \Delta_j - \Bigl(3n \cdot \alpha^3 + 2 \cdot \alpha^2  + (m + j) t \cdot \alpha + m t\Bigr)\\
& = \Delta_{j-1} + \delta - \Bigl(3n \cdot \alpha^3 + 2 \cdot \alpha^2 + (m+j) t \cdot \alpha + m t\Bigr)\\
& = \Delta_j^* + 2n\cdot \alpha^3 + \alpha^2 + 2j t \cdot \alpha + j t - \Bigl(3n \cdot \alpha^3+ 2 \cdot \alpha^2  + (m+j) t \cdot \alpha + m t\Bigr)\\
& = \Delta_j^* - n \cdot \alpha^3 - \alpha^2 -  (m - j) t \cdot \alpha - (m -j) \cdot t\\
        & = \Delta_{j-1} + 2n \cdot \alpha^3 + \alpha^2 + (2m -2j + 1) t \cdot \alpha + (m -j )t - n \cdot \alpha^3 - \alpha^2 - (m-j)t \cdot \alpha - (m-j) t\notag\\
        & = \Delta_{j-1} + n \cdot \alpha^3 + (m-j+1)t \cdot \alpha\,.\qedhere
\end{align*}
\end{proof}

\subsection{Candidate Sets}
\label{sec:candidates}

We will now consider a special kind of subsets of jobs called \emph{candidate sets} that encode solutions to the given \textsc{3-Partition} instance, as discussed in Section~\ref{sec:construction}. As we will later see in \Cref{sec:correctness}, it suffices to consider canonical schedules for candidates sets, since whenever there exists a feasible schedule with at most $k$ tardy jobs, there also exists such a schedule which is canonical for some candidate set (Lemma~\ref{lem:only-candidate}). We start by defining a candidate set of jobs:

\begin{definition}
A \emph{candidate set} is a set $\widetilde{\mathcal{J}}$ of jobs such that for each $i \in [n]$ and $j\in [m]$, exactly one of~$J^*_{i,j}$ and $\neg J^*_{i,j}$, and exactly one of $J_{i,j}$ and $\neg J_{i,j}$ is contained in $\widetilde{\mathcal{J}}$, and every filler or delimiter job is contained in $\widetilde{\mathcal{J}}$.

\end{definition}
Whenever we will talk in the following about the canonical schedule for a candidate set $\widetilde{\mathcal{J}}$, the possible tie between $F_i^1$ and $J_{i,1}^*$ will always be broken in favor of $J_{i,1}^*$ i.e., if both $F_i^1$ and $J_{i,1}^*$ are contained in $\widetilde{\mathcal{J}}$, then $F_i^1$ will be scheduled after $J_{i,1}^*$ in the canonical schedule for $\widetilde{\mathcal{J}}$.
We remark that not all candidate sets lead to feasible schedules.
\begin{example}
Assume that $n =4$, $m= 2$, and 
$$
\widetilde{\mathcal{J}}  = \mathcal{J}_{F} \cup \mathcal{J}_D \cup\{J^*_{1,1}, \neg J^*_{2,1}, \neg J^*_{3,1}, J^*_{4,1}, J^*_{1,2}, J^*_{2,2 }, J^*_{3, 2}, J^*_{4,2}\} \cup \{J_{1,1}, \neg J_{2,1}, \neg J_{3,1}, J_{4,1}, J_{1,2}, J_{2,2 }, J_{3, 2}, J_{4,2}\} 
$$
where $\mathcal{J}_F$ is the set of filler jobs and $\mathcal J_D$ is the set of delimiter jobs.
Then the canonical schedule for~$\widetilde{\mathcal J}$ is
\[\begin{array}{rrrrrrrrrrrrrrrrrrrr}
F_0, & J^*_{1,1} , & F_1^1, &  \neg J^*_{2, 1},  & F_2^1, &  \neg J^*_{3,1}, &  F_3^1, & J^*_{4,1} , & F_4^1, & D^*_1\\
& \neg J^*_{1,1} , & J_{1,1}, &  \neg J_{2, 1},  & J^*_{2,1}, &  \neg J_{3,1}, &  J^*_{3, 1}, & \neg J^*_{4,1} , & J_{4,1},  & D_1\\
\hline
& \neg J_{1,1}, & J^*_{1, 2} , &  J^*_{2, 2},  & J_{2,1}, &  J^*_{3,2}, & J_{3,1},
& \neg J_{4,1}, & J^*_{4, 2} , & D^*_2\\
& \neg J^*_{1,2} , & J_{1,2}, &  \neg J^*_{2, 2},  & J_{2,2}, &  \neg J^*_{3,2}, &  J_{3, 2},
& \neg J^*_{4,2} , & J_{4,2}, &  D_2\\
\hline
& \neg J_{1,2}, & F_1^m, &\neg J_{2,2}, & F_2^m, &\neg J_{3,2}, & F_3^m, &\neg J_{4,2}, &F_4^m\,.
\end{array}\]
The horizontal lines separate the time periods, and each row
represents one half of a period.
\end{example}

Canonical schedules for candidate sets are the focal point of our construction. In what follows we derive important structural properties of such schedules that will prove useful later on. We begin by first classifying which jobs are early in the canonical schedule of some candidate set (\Cref{lem:early-delimiter-jobs,lem:early-filler-jobs,lem:early-number-jobs}). We then give a necessary and sufficient condition for the feasibility of such a schedule, i.e., when do such schedules have maximum tardiness at most~$\ell$ (\Cref{lem:candidate->feasible}).
All of this will allow us to easily classify when a candidate set leads to a feasible schedule, and also help us in calculating how many tardy jobs such a schedule has.

We begin by proving that all filler jobs are early in the canonical schedule of any candidate set. We then categorize when number jobs and delimiter jobs are early in such schedules.  

\begin{lemma}\label{lem:early-filler-jobs}
In the canonical schedule for any candidate set, all filler jobs are early.
\end{lemma}

\begin{proof}
Job~$F_0$ is the first job and thus completed at time $p(F_0) = d(F_0)$. For each $i \in [n]$, job~$F_i^1$ completes after the filler jobs~$\{F_0\} \cup \{F_{i_0}^1 :i_0 \le i\}$ are completed, along with an additional $i$ number jobs. Thus, job~$F_i^1$ completes by time
$$
p(F_0) + i \cdot p(F_i^1) + i \cdot \max_{J} p(J) < 2i \cdot \alpha^3 + 0.1 \cdot \alpha^2 = d(F_i^1)\,.
$$
Next consider job~$F_i^m$ for $i \in [n]$. This job is completed after all jobs in~$\mathcal{J}^*_{\le m}$ are completed, all jobs in $\{J_{i_0, m},\neg J_{i_0, m} : i_0 \le i\}$ are completed, exactly one job in~$\{J_{i_0, m}, \neg J_{i_0, m}\}$ is completed for each $i_0 > m$ is completed, and the delimiter job $D_m$ and all filler jobs $\{F_{i_0}^m:i_0 \le i\}$ are completed. Thus, using \Cref{obs:deadline}, job~$F_i^m$ completes by time
\begin{align*}
p(\mathcal J^*_{\le m}) & + \sum_{i_0 = 1}^i\bigl(p(J_{i_0, m}) + p(\neg J_{i_0, m})\bigr) + \sum_{i_0=i+1}^n p(J_{i_0, m}) + p(D_m) + \sum_{i_0 = 1}^i p(F_{i_0}^m)\\
& \le \Delta_m - n \cdot  \alpha^3 - \alpha^2 - 2mt \cdot \alpha + (n + i ) \cdot \alpha^3  + mt\cdot \alpha + \alpha^2 + mt\cdot \alpha + i \cdot \alpha^3\\
& = \Delta_m+ 2i \cdot \alpha^3  < d(F_i^m)\,.\qedhere
\end{align*}
\end{proof}

\begin{lemma}
\label{lem:early-number-jobs}%
In the canonical schedule for a candidate set $\widetilde{\mathcal{J}}$, a number job~$J$ is early if and only if $J \in\widetilde{\mathcal{J}}$.
\end{lemma}

\begin{proof}
Let $i \in [n]$ and $j \in [m]$. We prove the lemma only for the number job $J^*_{i, j}$, as the proof for~$\neg J^*_{i, j}$, $J_{i, j}$, and $\neg J_{i, j}$ follows similar arguments and calculations. Assume first that $J^*_{i,j}\in\widetilde{\mathcal J}$. If $j=1$, then $J^*_{i,1}$ is scheduled after all jobs in $\{F_0\} \cup \{F_{i_0} :i_0 < i\}$, and exactly one of job in $\{J^*_{i_0,1},J^*_{i_0, 1}\}$ for each~$i_0 \le i$. Thus, $J^*_{i,1}$ is completed by time
\[
p(F_0) + \sum_{i_0 =1}^{i-1} p(F_{i_0}) + \sum_{i_0=1}^i p(\neg J^*_{i_0, 1})  = mt \cdot \alpha + (i-1) \cdot \alpha^3 + i \cdot \alpha^3 + \sum_{i_0 =1}^i a_i \le 2i \cdot \alpha^3 < d(J^*_{i,1})\,.
\]
Otherwise, if $j> 1$, job $J^*_{i,j}$ is scheduled after $\mathcal J^*_{\le j-1}$, $n+ i-1$ or $n+i$ jobs from $\mathcal{J}_{j-1}$ (depending on whether or not $J_{i,j-1} \in \widetilde{\mathcal J}$), $D_{j-1}$, and $i $ jobs from $\mathcal{J}^*_j$ (including $J^*_{i,j}$). The jobs from $\mathcal{J}_{j-1}$ each have a processing time at most $\alpha^3 + t \cdot \alpha$, while the jobs from $\mathcal{J}^*_j$ each have a processing time of at most $\alpha^3 + t$.  Consequently, using \Cref{obs:deadline}, $J^*_{i,j}$ is completed by time
\begin{align*}
p(\mathcal{J}^*_{\le j-1}) + &(n+i) \cdot (\alpha^3+ t \cdot \alpha )  + p(D_{j-1}) + i \cdot(\alpha^3+ t) \\
& = \Delta_{j-1} - n\cdot \alpha^3 - \alpha^2 - 2(j - 1)t \cdot \alpha + (n+2i)\cdot \alpha^3 + \alpha^2 + (n+i+j-1) t \cdot \alpha +it\\
&= \Delta_{j-1} + 2i \cdot \alpha^3 + (n+i-j+1) t \cdot \alpha + it < d(J^*_{i,j})\,. 
\end{align*}

Next consider the case where $J^*_{i,j} \notin \widetilde{\mathcal{J}}$. Then $J^*_{i,j}$ is scheduled after all jobs from~$\mathcal J_{\le j-1}$, as well as $n+i$ jobs from~$\mathcal{J}^*_{j}$ (including $J^*_{i,j}$), $D^*_{j}$, and $i$ jobs from~$\mathcal{J}_j$. The jobs in $\mathcal{J}^*_{j}$ and $\mathcal{J}_j$ each have a processing time of at least $\alpha^3$, and so consequently, using \Cref{eq:Jj-1_Dj-1}, job $J^*_{i,j}$ is completed not before time
\begin{align*}
p (\mathcal J_{\le j-1}) + & (n+i) \cdot \alpha^3 +p(D_j^*) + i \cdot \alpha^3 \\
& = \Delta_{j-1} + n \cdot \alpha^3 + (m-j+1) t \cdot \alpha + (n+2i) \cdot \alpha^3 + \alpha^2 + (m -j)t \cdot \alpha \\
& = \Delta_{j-1} + (2n + 2i) \cdot \alpha^3 + \alpha^2 + (2m-2j+1) t \cdot \alpha > d(J^*_{i,j}) \,.
\end{align*}
\end{proof}

\begin{lemma}\label{lem:early-delimiter-jobs}
In the canonical schedule for some candidate set~$\widetilde{\mathcal J}$, job $D_j^*$ is early if and only if ${\sum_{i : J^*_{i,j}\in \widetilde{\mathcal J}} a_i \ge jt}$.
Furthermore, $D_j$ is early if and only if $\sum_{i : J_{i,j}\in \widetilde{\mathcal J}} a_i \le jt$.
\end{lemma}

\begin{proof}
Let $j \in [m]$, and consider the delimiter job~$D_j^*$. This job is scheduled immediately after all jobs in~$\mathcal J_{\le j -1}$, followed by all jobs in $\{J^*_{i,j} : J^*_{i,j} \in \widetilde{\mathcal J}\} \cup  \{\neg J^*_{i,j} : J^*_{i,j} \notin \widetilde{\mathcal J}\}$. Thus, using \Cref{obs:ptildeJ}, the completion time of $D_j^*$ is given by
\begin{align*}
p(\mathcal J_{\le j -1}) &+ \sum_{i: J^*_{i,j} \in \widetilde{\mathcal J}} p(J^*_{i, j}) + \sum_{i : \neg J^*_{i,j}\in \widetilde{\mathcal J}} p(\neg J^*_{i,j})  + p(D_j^*) \\
& = \Delta_j^* - n \cdot \alpha^3 - \alpha^2 - (m-j)t \cdot \alpha - (m-j) t + n\cdot \alpha^3 + \sum_{i: \neg J^*_{i,j} \in \widetilde{\mathcal J}} a_i + \alpha^2 + (m-j)t\cdot \alpha\\
& = d(D_j^*) + \sum_{i:\neg  J^*_{i,j} \in \widetilde{\mathcal J}} a_i  - (m-j) t\,.
\end{align*}
Consequently, $D_j^*$ is early if and only if $\sum_{i:\neg J^*_{i,j} \in \widetilde{\mathcal J}} a_i  \le (m-j)t$. Since $\sum_{i=1}^n a_i = m t$, this is equivalent to $\sum_{i :J^*_{i,j}\in \widetilde{\mathcal J}} a_i  \ge jt$.

Next consider job $D_j$. This job is scheduled immediately after all all jobs in~$\mathcal{J}^*_{\le j}$, followed by all jobs in $\{J_{i,j} : J_{i,j} \in \widetilde{\mathcal J}\} \cup \{\neg J_{i,j} : \neg J_{i,j} \in \widetilde{\mathcal J}\}$, and
$D_j$. Using \Cref{obs:deadline}, the total processing time of these jobs is
\begin{align*}
p(\mathcal{J}^*_{\le j}) &+ \sum_{i: J_{i,j} \in \widetilde{\mathcal J}} p(J_{i, j}) + \sum_{i : \neg J_{i,j}\in\widetilde{ \mathcal J}} p(\neg J_{i,j})  + p(D_j)\\
& = \Delta_j - n \cdot \alpha^3 - \alpha^2 - 2jt \cdot \alpha + n\cdot \alpha^3 + \sum_{i :J_{i,j}\in \widetilde{\mathcal J}} a_i \cdot \alpha + \alpha^2 + jt\cdot \alpha\\
& = d(D_j) + \sum_{i : J_{i,j}\in\widetilde{ \mathcal J }} a_i \cdot \alpha  - j t \cdot \alpha\,.
\end{align*}
Consequently, $D_j$ is early if and only if $\sum_{i:J_{i,j} \in \widetilde{\mathcal J}} a_i  \le jt$.
\end{proof}

As discussed in Section~\ref{sec:construction}, we wish to encode the situation where an element $a_i$ of the \textsc{3-Partition} instance is selected to the one of the sets $S_1,\ldots,S_j$ by scheduling~$J_{i,j}$ and $J^*_{i,j}$ early, while the situation where $a_i \notin S_1 \cup \ldots \cup S_j$ should be encoded by scheduling~$\neg J_{i,j}$ and~$\neg J^*_{i,j}$ early. \Cref{lem:early-number-jobs} ensures that in any canonical schedule for some candidate set $\widetilde{\mathcal J}$, exactly one of~$\{J^*_{i,j},\neg J^*_{i,j}\}$ and one of $\{J_{i,j},\neg J_{i,j}\}$ are early for each $i \in [n]$ and $j \in [m]$, since these are contained in $\widetilde{\mathcal{J}}$ by definition. Note that this still allows, e.g., for both~$J^*_{i,j}$ and $\neg J_{i,j}$ to be early, or both $J^*_{i,j}$ and $\neg J^*_{i,j+1}$ to be early. The following lemma shows that such inconsistencies in the way we wish to encode solutions with candidate sets is avoided once considering schedules with maximum tardiness at most~$\ell$. 

\begin{lemma}\label{lem:candidate->feasible}
    Let $\widetilde{\mathcal J}$ be a candidate set.
    Then the canonical schedule for $\widetilde{\mathcal J}$ is feasible if and only if there is no $i \in[n]$ and $j\in [m-1]$ such that $\neg J^*_{i,j+1} ,J_{i,j}\in \widetilde{\mathcal{J}} $ and there is no $i\in [n]$ and $j\in [m]$ such that $\neg J_{i,j} , J^*_{i,j}\in \widetilde{\mathcal J}$.
\end{lemma}

\begin{proof}
By \Cref{lem:early-filler-jobs}, all filler jobs are early and thus have tardiness at most $\ell$, independent on the candidate set $\mathcal{\widetilde{J}}$. We continue with the number jobs $J^*_{i, j}$, $\neg J^*_{i,j}$, $J_{i,j}$, and $\neg J_{i,j}$.
    Note that by \Cref{lem:early-number-jobs}, for each number job~$J$ we only need to consider the case that $J \notin \mathcal{\widetilde{J}}$ (as $J$ is early otherwise).
    
    We start with showing that job~$J^*_{i,j}$ for $i\in [n]$ and $j \in [m]$ always has tardiness at most $\ell$.
    By our previous observation, it suffices to consider the case $J^*_{i,j} \notin \mathcal{\widetilde{J}}$.
    Before $J^*_{i,j}$, all jobs from~$\mathcal J_{\le j-1}$ as well as $n+i$ jobs from~$\mathcal{J}^*_{j}$ (including $J^*_{i,j}$, and each with processing time at most $ \alpha^3 + t$), $D^*_{j}$, and $i$ jobs (each with processing time at most $p (J) \le \alpha^3 + t  \cdot \alpha$) from~$\mathcal{J}_j$ are scheduled.
    Using \Cref{eq:Jj-1_Dj-1}, $J^*_{i,j}$ is completed by time
    \begin{align*}
        p (\mathcal J_{\le j-1}) + & (n+i) \cdot (\alpha^3 + t) + \alpha^2 + (m-j)t \cdot \alpha + i \cdot (\alpha^3 + t\cdot \alpha)\\
        & = \Delta_{j-1} + n \cdot \alpha^3 + (m-j+1)t \cdot \alpha + (n+2i) \cdot \alpha^3  + \alpha^2 + (m-j+ i)t \cdot \alpha + (n+i)t\\
        & = \Delta_{j-1} + (2n +2i) \cdot \alpha^3 + \alpha^2 + (2m -2j + i + 1) t \cdot \alpha + (n + i) t\\
        & < \Delta_{j-1} + 2(n+i) \cdot \alpha^3 + \alpha^2 + 0.1 \cdot \alpha^2  = d(J^*_{i,j}) + 2n \cdot \alpha^3 + \alpha^2 < d(J^*_{i,j}) + \ell\,.
    \end{align*}
    
    We continue by characterizing when job~$\neg J^*_{i,j}$ for $i\in[n]$ and $j \in [m]$ has tardiness at most $\ell$.
    Again, we only consider the case that $\neg J^*_{i,j} \notin \mathcal{\widetilde{J}}$.
    In this case, we will show that job $\neg J^*_{i,j}$ has tardiness at most~$\ell$ if and only if $J_{i,j} \in \widetilde{ \mathcal{J}}$.
    Before $\neg J^*_{i,j}$, all jobs from~$\mathcal J_{\le j-1}$ as well as $n + i$ jobs from $\mathcal{J}^*_{j}$ (including~$\neg J^*_{i,j}$ itself), $D_{j}^*$, and $i$ or $i-1$ jobs from $\mathcal{J}_j$ (depending on whether $J_{i,j} \in \widetilde{\mathcal J}$) are scheduled.
    If $J_{i,j} \in \widetilde{\mathcal J}$, then there are only $i-1$ jobs from $\mathcal{J}_j$ before $\neg J^*_{i,j}$ and consequently,
    using analogous calculations as for $J^*_{i,j}$ and \Cref{eq:Jj-1_Dj-1}
    (note that the only difference to $J^*_{i,j}$ is that there are now $i-1$ instead of $i $ jobs from $\mathcal{J}_j$), job~$\neg J^*_{i,j}$ is completed by time 
    \begin{align*}
        p &(\mathcal J_{\le j-1})  + (n+i) \cdot (\alpha^3 + t) + \alpha^2 + (m-j)t \cdot \alpha + (i -1) \cdot (\alpha^3 + t\cdot \alpha)\\
        & = \Delta_{j-1} + n \cdot \alpha^3 + (m-j+1)t \cdot \alpha + (n + 2i - 1) \cdot \alpha^3 + \alpha^2 + (m-j + i-1) t \cdot \alpha + (n +i) t\\
        & = \Delta_{j-1} + (2n + 2i -1) \cdot \alpha^3 + \alpha^2 + (2m -2j + i) t \cdot \alpha + (n+i) t\\
        & < \Delta_{j-1} + 2(n+i-1) \cdot \alpha^3 + \alpha^2 + 0.1 \cdot \alpha^2  = d(\neg J^*_{i,j}) + 2n \cdot \alpha^3 + \alpha^2 < d(\neg J^*_{i,j}) + \ell\,.
    \end{align*}
    If, however, $\neg J_{i,j} \in \widetilde{\mathcal J}$, then there are $i$ jobs from $\mathcal{J}_j$ before $\neg J^*_{i,j}$ and thus $\neg J^*_{i,j}$ is not completed before (using \Cref{eq:Jj-1_Dj-1} and analogous calculations as for $J^*_{i,j}$)
    \begin{align*}
        p (\mathcal J_{\le j-1}) + & (n+i) \cdot \alpha^3 + \alpha^2 + (m-j)t \cdot \alpha + i \cdot \alpha^3\\
        & = \Delta_{j-1} + n \cdot \alpha^3 + (m-j+1)t \cdot \alpha + (n+i) \cdot \alpha^3 + \alpha^2 + (m-j)t \cdot \alpha + i \cdot \alpha^3\\
        & = \Delta_{j-1} + (2n + i) \cdot \alpha^3 + \alpha^2 + (2m -2j + 1)t \cdot \alpha\\
        & = d(\neg J^*_{i,j}) + (2n + 1) \cdot \alpha^3 + 0.9 \cdot \alpha^2 + (2m-2j + 1)t \cdot \alpha
         > d  (\neg J^*_{i,j}) + \ell\,.
    \end{align*}
    Consequently, $\neg J^*_{i,j } $ is has tardiness larger than $\ell$ if and only if $ J^*_{i,j} \in \widetilde{\mathcal{J}}$ and~$\neg J_{i, j} \in \mathcal{\widetilde{J}}$.

    The arguments for $J_{i, j}$ and $\neg J_{i,j}$ are analogous:
    To show that $J_{i,j}$ has tardiness at most $\ell$, we again only consider the case $J_{i,j} \notin \mathcal{\widetilde{J}}$.
    Before $J_{i,j}$, all jobs from $\mathcal{J}_{\le j}^*$ as well as $n+i$ jobs from~$\mathcal{J}_{j}$ (including $J_{i,j}$), $D_{j}$, and $i$ jobs from~$\mathcal{J}^*_{j+1}$ are scheduled.
    Thus, using \Cref{eq:tildeJ_Delta*}, $J_{i, j}$ is completed by time
    \begin{align*}
        p(\mathcal{J}_{\le j}^*)  + &(n +i) \cdot (\alpha^3 + t \cdot \alpha) + \alpha^2 + jt \cdot \alpha + i \cdot (\alpha^3  + t)\\
        & = \Delta_j^* + n \cdot \alpha^3 + jt + (n + 2i) \cdot \alpha^3 + \alpha^2 + (n+i+j) \cdot \alpha + it\\
        & < \Delta_j^* + 2(n +i) \cdot \alpha^3 + \alpha^2 +0.1 \cdot \alpha^2 < d(J_{i,j}) + \ell.
    \end{align*}
    
    We continue with job $\neg J_{i,j}$, again only considering the case that $\neg J_{i, j} \notin \widetilde{\mathcal J}$.
    Note that there are $i-1$ jobs from $\mathcal{J}^*_{j+1}$ scheduled before~$\neg J_{i,j}$ if $J^*_{i,j+1} \in \widetilde{\mathcal J}$, and $i $ jobs from $\mathcal{J}^*_{j+1}$ otherwise.
    Again using \Cref{eq:tildeJ_Delta*}, we get that if $ J^*_{i, j+1} \in \mathcal{\widetilde{J}}$, then $\neg J_{i,j}$ is completed not before time (using \Cref{eq:tildeJ_Delta*})
    \begin{align*}
        p(\mathcal{J}^*_{\le j}) & + (n+i) \cdot \alpha^3 + \alpha^2  + (m-j)t \cdot \alpha + i \cdot \alpha^3\\
        & = \Delta_j^* + n \cdot \alpha^3 + jt + (n + 2i) \cdot \alpha^3 + \alpha^2 + (m-j)t \cdot \alpha\\
        & > \Delta_j^* + (2n + 2i) \cdot \alpha^3 > d(\neg J_{i,j}) + \ell
    \end{align*}
    and thus $\neg J_{i,j}$ is has tardiness at most~$\ell$.
    If, on the other hand we have $\neg J^*_{i,j+1} \in \mathcal{\widetilde{J}}$, then $\neg J_{i,j}$ is completed by time (using \Cref{eq:tildeJ_Delta*})
    \begin{align*}
        p(\mathcal{J}^*_{\le j}) & + (n+i) \cdot (\alpha^3 + t \cdot \alpha) + \alpha^2  + (m-j)t \cdot \alpha + (i - 1) \cdot (\alpha^3 + t)\\
        & = \Delta_j^* + n \cdot \alpha^3 + jt + (n + 2i-1) \cdot \alpha^3 + \alpha^2 + (n + i + m-j)t \cdot \alpha + (i-1)t\\
        & < \Delta_j^* + (2n + 2i-1) \cdot \alpha^3 + 1.1 \cdot \alpha^2 < d(\neg J_{i,j}) + \ell
    \end{align*}
    and thus $\neg J_{i,j}$ has tardiness at most $\ell$.
    Combining the two statements, $\neg J^*_{i,j}$ has tardiness larger than $\ell$ if and only if $ J^*_{i,j} \in \mathcal{\widetilde{J}}$ and $\neg J^*_{i,j+1} \in \mathcal{\widetilde{J}}$.

    We continue with job~$D_j^*$.
    This job is scheduled after $\mathcal J_{\le j-1}$ and $n $ jobs from~$\mathcal{J}^*_j$.
    Thus, using \Cref{obs:ptildeJ}, it is completed by time
    \begin{align*}
        p(\mathcal{J}^*_j) + & n \cdot (\alpha^3 + t) + \alpha^2 + (m-j)t\\
        & = \Delta_j^* - n \cdot \alpha^3 - \alpha^2 -(m-j)t \cdot \alpha - (m-j) t +n \cdot (\alpha^3 + t) +  \alpha^2 + (m-j)t \cdot \alpha\\
        & = \Delta_j^* + (n+j-m) \cdot t < d(D_j^*) + \ell.
    \end{align*}
    Similarly, $D_j$ is scheduled after $\mathcal{J}_{\le j}^*$ and $n$ jobs from $\mathcal{J}_j$ and thus (using \Cref{obs:deadline}) completed by time
    \begin{align*}
        p(\mathcal{J}_j) + &n \cdot (\alpha^3 + t \cdot \alpha) + \alpha^2 + jt\\
        & = \Delta_j -n \cdot \alpha^3 - \alpha^2 - 2jt \cdot \alpha + n \cdot (\alpha^3 + t \cdot \alpha) + \alpha^2 + j t \cdot \alpha\\
        & = \Delta_j + (n-j)t \cdot \alpha < d(D_j ) + \ell.\qedhere
    \end{align*}
\end{proof}

\subsection{Correctness}
\label{sec:correctness}

In order to show the correctness of the reduction described in this section, the first (and hardest) part is to show that we can restrict ourselves to considering canonical schedules.
This will be done in the next lemma:
\begin{restatable}[]{lemma}{LemCandidate}\label{lem:only-candidate}
    If there is a feasible schedule with $k$ tardy jobs, then there is a feasible schedule with at most $k$ jobs that is the canonical schedule for a candidate set.
\end{restatable}

Before proving \Cref{lem:only-candidate} in \Cref{sec:lem5proof}, we first show how to derive the main result in this section, namely the strong NP-completeness of $1|\sum U_j \le k, T_{\max} \le \ell|$, from \Cref{lem:only-candidate}:

\begin{proof}[Proof of \Cref{thm:strong}]
    $1|T_{\max} \le \ell, U_j \le k|$ is clearly contained in NP as a feasible schedule with at most $k$ tardy jobs is a certificate.
    It remains to show that the problem is NP-hard.
    We will do so via the reduction from \textsc{3-Partition} as described throughout this section.
    The reduction clearly runs in polynomial time, so it remains to show its correctness.
    
    We start with the forward direction.
    So let $(S_1, \ldots, S_m)$ be a solution to the \textsc{3-Partition} instance.
    Let 
    \begin{align*}
        \widetilde{\mathcal{J}} := & \{J^*_{i,j}, J_{i,j} : a_i \in S_1 \cup \ldots \cup S_j\} \cup \{ \neg J^*_{i,j},\neg J_{i,j}:a_i \notin S_1 \cup \ldots \cup S_j\} \\
        & \cup \{F_0, F_i^1, F_i^m: i \in [n]\} \cup \{D_j, D_j^*: j \in [m]\},
    \end{align*}
    and let $\sigma$ be the canonical schedule for $\widetilde{\mathcal{J}}$.
    By \Cref{lem:candidate->feasible}, $\sigma $ is feasible.
    By \Cref{lem:early-filler-jobs}, jobs~$F_0$, $F_i^1$, and $F_i^m$ for $i\in[n]$ are early.
    By \Cref{lem:early-number-jobs}, we have $2nm$ early and $2nm $ tardy number jobs.
    By \Cref{lem:early-delimiter-jobs} and since we have a solution to \textsc{3-Partition}, all delimiter jobs are early.
    Overall, we have $k = 2mn$ tardy jobs, finishing the forward direction.

    We continue with the backward direction.
    So let $\sigma$ be a feasible schedule with at most $k$ tardy jobs.
    By \Cref{lem:only-candidate}, we may assume that $\sigma$ is the canonical schedule for some candidate set $\widetilde{\mathcal J}$.
    By \Cref{lem:early-number-jobs}, there are $2nm$ tardy number jobs.
    This implies that $2nm$ number jobs and all non-number jobs are early.
    Let $I_j^* := \{i \in [n]: J^*_{i,j}\in \widetilde{\mathcal{J}}\}$ and $I_j := \{i\in[n] : J_{i,j} \in \widetilde{\mathcal J}\}$.
    Since $\sigma$ is feasible, \Cref{lem:candidate->feasible} implies that whenever $J_{i,j}^* \in \mathcal{\widetilde{J}}$, then also $J_{i,j} \in \mathcal{\widetilde{J}}$ (as otherwise $\neg J_{i,j}, J_{i,j}^* \in \widetilde{\mathcal J}$, contradicting the feasibility of $\sigma $ by \Cref{lem:candidate->feasible}).
    In other words, we have $I_j^* \subseteq I_j$.
    Similarly, the feasibility of $\sigma$ together with \Cref{lem:candidate->feasible} implies that whenever $J_{i,j} \in \widetilde{\mathcal{J}}$, then also $J_{i,j + 1}^* \in \widetilde{\mathcal{J}}$, i.e., $I_j \subseteq I_{j+1}^*$.
    Because every delimiter job is early, \Cref{lem:early-delimiter-jobs} implies that $\sum_{i \in I_j^*} a_i \ge jt$ and $\sum_{i\in I_j} a_i \le jt$.
    Combining the above two results, we have
    \[
        \sum_{i \in I_j^*} a_i  \ge jt \ge  \sum_{i \in I_j} a_i \,.
    \]
    Since $I_j^* \subseteq I_j$, this implies that $I_j^* = I_j$ and that the above inequality holds with equality.
    We claim that $(S_1, \ldots, S_m) := (I_1^*, I_2^*\setminus I_1^*, \ldots, I_m^* \setminus I_{m-1}^*)$ is a solution to the \textsc{3-Partition} instance.
    First note that $S_{j_0}$ and $S_{j_1}$ are disjoint for $j_0 \neq j_1$ as $I_1^* = I_1 \subseteq I_2^* = I_2 \subseteq I_3^* = I_3 \subseteq\ldots \subseteq I_m^* \subseteq I_m$.
    Furthermore, we have
    \[
        \sum_{i \in S_j } a_i = \sum_{i\in I_j^*} a_i - \sum_{i \in I_{j-1}^*} a_i = jt- (j-1)t = t\,.\qedhere
    \]
\end{proof}

\subsection{Proof of Lemma~\ref{lem:only-candidate}}
\label{sec:lem5proof}

We conclude this section with the proof of \Cref{lem:only-candidate}. That is, we show that if there exists a feasible schedule with at most $k$ tardy jobs then there must be such a schedule which is the canonical schedule of some candidate set. Recall that \Cref{obs:almost-edd} implies that there always is a feasible schedule with a minimum number of tardy jobs which is the canonical schedule of some subset $\mathcal{\widetilde{J}}$ of jobs. We show that if the number of tardy jobs in the feasible canonical schedule of some subset~$\mathcal{\widetilde{J}}$ of jobs is at most $k$, then $\mathcal{\widetilde{J}}$ must be a candidate set. In order to so, we start with very coarse statements on feasible canonical schedules, and then refine them step by step.
We start with the following very coarse structure:
\begin{lemma}\label{lem:first-restriction}
Any feasible canonical schedule~$\sigma$ satisfies that
\begin{itemize}
\item for every $j \in [m]$, jobs from~$\mathcal{J}^*_j$ and $D_j^*$ are scheduled after every job from~$\mathcal{J}^*_{\le j - 1}$ and before any job not contained in $\mathcal{J}_{\le j}$, and
\item for every $j \ge 2$, jobs from $\mathcal{J}_j$ and $D_j$ are scheduled after every job from~$\mathcal J_{\le j-1}$ and before every job not contained in~$\mathcal J^*_{\le j+1}$.
\end{itemize}
\end{lemma}

\begin{proof}
Suppose $\sigma$ is canonical for some subset of jobs $\mathcal{\widetilde{J}}$, and let $\widetilde{d}(J)$ denote the modified due date of any job $J$ as defined in Definition~\ref{def:canonical}. We start with the first bullet point. Note that the maximal due date in $\mathcal{J}^*_{\le j-1}$ is~$d(D^*_{j-1})$. Thus, for each job~$J \in \mathcal{J}^*_{\le j-1}$, we have
\begin{align*}
 \widetilde{d}(J)  & \le d(D^*_{j-1}) + \ell \\
& = \Delta_{j-1}^* + 2n \cdot \alpha^3 + 1.1 \cdot \alpha^2 \\
& = \Delta_{j-2} + \delta_{j-1}^* + 2n \cdot \alpha^3 + 1.1 \cdot \alpha^2\\
& = \Delta_{j-2} + 2n \cdot \alpha^3 + \alpha^2 + (2m-2j+3) t \cdot \alpha + (m-j+1) t+ 2n \cdot \alpha^3 + 1.1 \cdot \alpha^2  \\
& = \Delta_{j-2} + 4n \cdot \alpha^3 + 2.1 \cdot \alpha^2 + (2m -2j + 3) t  \cdot \alpha + (m -j+1) t\\
& = \Delta_{j-1} + 0.1 \cdot \alpha^2 + (2 - 2j) t \cdot \alpha + (1-j) \cdot t < d (\neg J^*_{i,j}) < d(J^*_{i,j}) < d (D_j^*)\,.
\end{align*}
Furthermore, for any $J \notin \mathcal{J}_{\le j}$, we have
$$
\max\{\widetilde{d}(\neg J^*_{i,j}), \widetilde{d}(J^*_{i,j}), \widetilde{d}(D_j^*)\}  \le d (D^*_{j}) + \ell = \Delta_j^* + \ell = \Delta_{j } + 0.1 \cdot \alpha^2 - 2jt \cdot \alpha - j t < d(J)\,. 
$$
Thus, the first bullet point follows from the fact that $\sigma$ is canonical.

The proof of the second bullet point is similar to the one of the first bullet point. For each $J \in \mathcal{J}_{\le j-1}$, we have
\begin{align*}
\widetilde d(J) & \le d(D_{j-1}) + \ell \\
& = \Delta_{j-1} + 2n \cdot \alpha^3 + 1.1 \cdot \alpha^2 \\
& = \Delta_{j-1} + \delta_j^* + 0.1 \cdot \alpha^2 - (2m-2j+1 ) \cdot \alpha - (m- j) t\\
& = \Delta_j^* + 0.1 \cdot \alpha^2 - (2m-2j+1 ) \cdot \alpha - (m- j) t < d(\neg J_{i,j}) < d(J_{i, j}) < d( D_j)\,.
\end{align*}
Furthermore, for any $J \notin \mathcal{J}^*_{\le j+1}$, we have
$$
\max\{\widetilde d(\neg J_{i,j}), \widetilde d(J_{i,j}), \widetilde d(D_j) \} \le d (D_{j}) + \ell = \Delta_j + \ell = \Delta_{j +1}^* + 0.1 \cdot \alpha^2 - (2m -2j -1)t \cdot \alpha - (m -j-1) t < d(J)\,. 
$$        
The second bullet point now follows from the fact that $\sigma $ is canonical.
\end{proof}

In the following, we mostly focus on the jobs with processing time at least $\alpha^3$, i.e., all jobs but the delimiter jobs and~$F_0$. Thus, we will call a job \emph{large} if its processing time is at least $\alpha^3$. All other jobs are called \emph{small}.
Therefore, the large jobs are precisely the number jobs and the filler jobs except for $F_0$, i.e., the set of large jobs is $\{J^*_{i,j}, \neg J^*_{i,j}, J_{i, j}, \neg J_{i, j}: i\in [n], j \in [m]\} \cup \{F_i^1, F_i^m : i\in [n]\}$. The small jobs are precisely the delimiter jobs and $F_0$, i.e., the set of small jobs is $\{D_j, D_j^* : j \in [m]\} \cup \{F_0\}$.

We now bound the number of early jobs from each set~$\mathcal{J}^*_j$ or $\mathcal{J}_j$.

\begin{lemma}\label{lem:bound-k}
For any feasible canonical schedule,
\begin{itemize}
\item up to $2n $ jobs from~$\{F_i^1, J^*_{i, 1}, \neg J^*_{i, 1}: i \in [n]\}$ can be completed by~$\Delta_1^*$,
\item up to $n $ jobs from~$\mathcal{J}^*_j$ can be completed by~$\Delta_j^*$ for any $j \ge 2$, and
\item up to $n$ jobs from~$\mathcal{J}_j$ can be completed by~$\Delta_j$ for $j\in[m]$.
\end{itemize}
\end{lemma}

\begin{proof}
Let $\sigma$ be a canonical feasible schedule.
        We first show the first bullet point.
        Note that each job from $\{F_i^1, J^*_{i, 1}, \neg J^*_{i, 1}: i \in [n]\}$ is large, i.e., has processing time at least~$\alpha^3$.
        Thus, at most $2n$ jobs from $\{F_i^1, J^*_{i, 1}, \neg J^*_{i, 1}: i \in [n]\}$ can be completed by~$\Delta_1^* < (2n +1) \cdot \alpha^3$.

        We continue with the second bullet point.
    By \Cref{eq:negtildeJ}, all jobs from~$\mathcal{J}_{\le j-1} \setminus \{D_{j-1}\}$ are completed by~$\Delta_j^*$.
        Thus, the total processing time of all jobs from~$\mathcal{J}^*_j$ that are completed by~$\Delta_j^* $ is upper-bounded by $\Delta_j^* - p(\mathcal{J}_{\le j-1} \setminus \{D_{j-1}\}) = n\cdot \alpha^3 + 2 \cdot \alpha^2 + (m-1) t \cdot \alpha + (m-j)t< (n + 1) \cdot \alpha^3$ using \Cref{obs:ptildeJ} for the equality.
        Because each job from~$\mathcal{J}^*_j$ has a processing time of at least~$\alpha^3$, this implies that at most~$n$ jobs from $\mathcal{J}^*_j$ can be completed by~$\Delta_j^*$.

        We conclude with the third bullet point which is analogous to the second bullet point.
        By \Cref{eq:negJ}, all jobs from~$\mathcal{J}^*_{\le j} \setminus \{D^*_{j}\}$ are completed by time~$\Delta_j$.
        Thus, the total processing time of all jobs from~$\mathcal{J}_j$ which are completed by time~$\Delta_j$ is upper-bounded by $\Delta_j- p(\mathcal{J}^*_{\le j} \setminus \{D^*_{j}\}) = n\cdot \alpha^3 + 2 \cdot \alpha^2 +  (m+j) t \cdot \alpha < (n + 1) \cdot \alpha^3$ using \Cref{obs:deadline} for the equality.
        Because each job from~$\mathcal{J}_j$ has processing time at least~$\alpha^3$, this implies that at most~$n$ jobs from $\mathcal{J}_j$ can be completed by time~$\Delta_j$.
    \end{proof}
    As an easy consequence of \Cref{lem:bound-k}, we get restrictions on the set of early jobs in an optimal solution:
    \begin{lemma}\label{claim:bound-k}
        For any canonical feasible schedule with at most $k$ tardy jobs, the set of tardy jobs looks as follows:
        \begin{itemize}
            \item $n $ jobs from~$\{F_i^1: i \in [n]\} \cup \mathcal{J}^*_1$ are tardy,
           and these tardy jobs are completed after $\Delta_1^*$,
            \item $n $ jobs from~$\mathcal{J}^*_j$ are tardy for any $j \ge 2$,
            and these tardy jobs are completed after $\Delta_j^*$, and
            \item $n$ jobs from~$\mathcal{J}_j$ are tardy for $j\in[m]$, and all these tardy jobs are completed after $\Delta_j$.
        \end{itemize}
    \end{lemma}
    \begin{proof}
        Note that each job from $\mathcal{J}^*_j$ that is early must be completed by time~$\Delta_j^*$.
        Further, each job from~$\mathcal{J}_j$ that is early must be completed by time~$\Delta_j$.
        If $F_i^1$ is early, then it is completed by time~$d(F_i^1) < \Delta_1^*$.
        Thus, \Cref{lem:bound-k} implies that for each $j \ge 2$, there are $2n$ tardy jobs from $\mathcal{J}^*_j \cup \mathcal{J}_j$, and there are $2n$ tardy jobs from $\mathcal{J}^*_1 \cup \mathcal{J}_1 \cup\{F_i^1 :i \in [n]\}$.
        Since there are at most $k =2mn$ tardy jobs overall, it follows that for each~$\mathcal{J}^*_j$ or~$\mathcal{J}_j$ with the exception of $\mathcal{J}^*_1$, there are exactly $n$ tardy jobs.
        
        The second part of each bullet point follows from \Cref{lem:bound-k} which implies that at most $n$ jobs from~$\mathcal{J}^*_j$, respectively $\mathcal{J}_j$, can be scheduled up to time $\Delta_j^* $, respectively $\Delta_j$.
    \end{proof}
    
We continue by showing strong restrictions on which large job may be scheduled in which time period:
\begin{lemma}
\label{lem:possible-jobs-in-halfperiod}%
Any canonical feasible schedule $\sigma$ with at most $k$ tardy jobs satisfies that the only large jobs which can be completed in
\begin{itemize}
\item the first half of the first period are $F_i^1$ for $i \in[n]$ and $\mathcal{J}^*_1 \cup \mathcal{J}_1$, and all these jobs are completed after time $\Delta_{j-1} + 0.1 \cdot \alpha^3$,
\item the first half of the $j$-th period for $j > 1$ are $\mathcal{J}^*_j \cup \mathcal{J}_{j-1}$, and all these jobs are completed after time $\Delta_{j-1} + 0.1 \cdot \alpha^3$, and
\item the second half of the $j$-th period are $\mathcal{J}^*_j \cup \mathcal{J}_{j}$, and all these jobs are completed after time $\Delta_{j}^* + 0.1 \cdot \alpha^3$.
\end{itemize}
\end{lemma}

\begin{proof}
    We start with the first bullet point.
    By \Cref{claim:bound-k}, $2n $ jobs from $\{F_i^1 : i \in [n]\} \cup \mathcal{J}^*_1$ are early and thus completed in the first half of the first period.
    These are the only large jobs completed in the first half of the first period as each large job has processing time at least $\alpha^3$ and $\Delta_1^* < (2n+1) \cdot \alpha^3$.

    We continue with the second bullet point.
    Let $j \ge 2$.
    Any large job~$J$ from~$\mathcal{J}^*_{\le j-1}$ must be completed by the end of the $(j-1)$-th period by \Cref{eq:negJ} and thus before the first half of the $j$-th period.
    As the schedule is feasible, any job from $\mathcal{J}_{j-1}$ is completed by time $\Delta_j^*$ by \Cref{eq:negtildeJ}.
    Further, $n$ jobs from $\mathcal{J}_j^*$ are early by \Cref{claim:bound-k} and thus completed by time $\Delta_j^*$.
    As $\Delta_j^* < \bigl((j-1) \cdot 4n  + 2n + 1\bigr) \cdot \alpha^3$, no other large job except from the $4(j-1) n + n$ large jobs from $\mathcal{J}_{\le j-1}$ and the $n$ early jobs from $\mathcal{J}_j^*$ are completed by time~$\Delta_j^*$.
    Consequently, all large jobs which can be scheduled in the first half of the $j$-th period are among~$J^*_{i,j}$, $\neg J^*_{i,j}$, $J_{i,j-1}$, and $\neg J_{i,j-1}$ for $i \in [n]$.

    Lastly, we show the third bullet point analogously to the second bullet point.
    Let $j \in [m]$.
    Any large job~$J$ from~$\mathcal{J}_{\le j-1}$ must be completed by the end of the first half of the $j$-th period by \Cref{eq:negtildeJ} and thus before the second half of the $j$-th period.
    As the schedule is feasible, any job from $\mathcal{J}^*_{j}$ is completed by time $\Delta_j$ by \Cref{eq:negJ}.
    Further, $n$ jobs from $\mathcal{J}_j$ are early by \Cref{claim:bound-k} and thus completed by time $\Delta_j$.
    As $\Delta_j < \bigl(j \cdot 4n + 1\bigr) \cdot \alpha^3$, no other large job except from the $4(j-1) n + 3n $ large jobs from $\mathcal{J}_{\le j}^*$ and the $n$ early jobs from $\mathcal{J}_j$ are completed by time $\Delta_j$.
    Consequently, all large jobs which can be scheduled in the second half of the $j$-th period are among $J^*_{i,j}$, $\neg J^*_{i,j}$, $J_{i,j}$, and $\neg J_{i,j}$ for $i \in [n]$.
\end{proof}

We continue with an observation which will be used in the next lemma.
\begin{observation}\label{obs:rounding}
    For each $x \in [4nm+2]$, there is one large job completed in the interval~$[x\cdot \alpha^3, (x+ 0.1) \cdot \alpha^3]$.
\end{observation}
\begin{proof}
    The observation follows from the fact that $\alpha$ was chosen large enough so that $0.1 \cdot \alpha^3$ is larger than all terms of order $\alpha^2$ or lower occurring in the processing times of the jobs together.
\end{proof}

    As a final step before proving \Cref{lem:only-candidate}, we show that exactly one of $J_{i,j} $ and $\neg J_{i, j}$ as well as exactly one of $J_{i, j}^*$ and $\neg J_{i, j}^*$ is early.
        \begin{lemma}\label{lem:upper-lower}
            If there exists a feasible canonical schedule with $k$ tardy jobs, then there is one where for each~$i \in [n]$ and $j \in [m]$,
            \begin{enumerate}[(i)]
                \item\label{item:at-least} at least one of $J_{i, j}$ and $\neg J_{i, j}$ and at least one of $J^*_{i, j}$ and $\neg J^*_{i, j}$ is tardy, and
               \item\label{item:at-most} at most one of $J_{i, j}$ and $\neg J_{i, j}$ and at most one of $J^*_{i, j}$ and $\neg J^*_{i, j}$ is tardy.
            \end{enumerate}
        \end{lemma}

        \begin{proof}
            Since~$F_i^1$ and $J^*_{i, 1}$ have the same characteristics, we may assume, without loss of generality, that if at least one of $F_i^1$ and $J^*_{i, 1}$ is early, then $F_i^1$ is early.
        
            We prove the lemma by backwards induction on~$i$, i.e., we will show that assuming that exactly one of~$J_{i_0, j}$ and $\neg J_{i_0, j}$ and exactly one of $J^*_{i_0, j}$ and $\neg J^*_{i_0, j}$ is tardy for every $i_0 > i$ and $j \in [m]$, also exactly one of $J_{i, j}$ and $\neg J_{i, j}$ and exactly one of $J^*_{i, j}$ and $\neg J^*_{i, j}$ is tardy for every $j \in [m]$.
            The proofs of \Cref{item:at-least,item:at-most} are analogous and both divided into three steps:
            First, we show that at least one of $J_{i, m}$ and~$\neg J_{i, m}$ (respectively at most one of $J^*_{i, 1} $ and $\neg J_{i, 1}^*$) is tardy.
            Second, we show that at least one of $J_{i, j}$ and $\neg J_{i, j}$ being tardy implies that at least one of $J_{i,j}^*$ and $\neg J_{i,j}^*$ is tardy (respectively at most one of $J^*_{i, j}$ and $\neg J^*_{i,j}$ being tardy implies that at most one of $J_{i,j}$ and $\neg J_{i,j}$ is tardy).
            Third, we show that at least one of $J^*_{i, j}$ and~$\neg J^*_{i, j}$ being tardy implies at least one of $J_{i, j-1}$ and~$\neg J_{i, j-1}$ is tardy (respectively at most one of $J_{i, j}$ and~$\neg J_{i,j}$ being tardy implies that at most one of $J^*_{i,j+1}$ and $\neg J^*_{i,j+1}$ is tardy).

            We continue by proving \Cref{item:at-least}, starting with the first step.
            After time~$T:= \Delta_m  + (2i -1.5) \cdot \alpha^3$, at least $2(n+1) - 2i$ jobs must be completed as $p(\mathcal{J}) - T >(2n + 2i - 1.5) \cdot \alpha^3$ and no job has processing time larger than $\alpha^3 + t \cdot \alpha$.
            The only jobs which can be completed after time~$T$ are $J_{i_0, m}$ and $\neg J_{i_0, m}$ for $i_0 \ge i$ (and in this case, these jobs are tardy) and $F_{i_0}^m$ for $i_0\in [m]$.
            For $i_0 < i$, job $F_{i_0}^m$ cannot be completed after time $T$ as $F_{i_0}^m$ is early for every $i_0 \in [n]$ by \Cref{claim:bound-k} and $d(F_{i_0}^m) < T$.
            Further, we have that exactly one of~$J_{i_0, m}$ and $\neg J_{i_0, m}$ for $i_0 > i$ is tardy.
            Consequently, apart from $J_{i,m}$ and $\neg J_{i,m}$, there are only $2n - 2i +1$ candidates to be completed after time~$T$, implying that at least one of $J_{i,m}$ and $\neg J_{i,m}$ is completed after time~$T$ and therefore tardy.

            We continue with the second step, showing that at least one of $J_{i, j}$ and $\neg J_{i, j}$ being tardy implies that at least one of $J_{i,j}^*$ and $\neg J_{i,j}^*$ is tardy.
        Consider the interval $I := [\Delta_j - (2n - 2i + 1.1) \cdot \alpha^3, \Delta_j+0.1\cdot \alpha^3]$.
        By \Cref{lem:possible-jobs-in-halfperiod}, the only large jobs which can be completed in this interval are $\mathcal{J}_j^* \cup \mathcal{J}_j$ and, if $j = 1$, also $\{F_{i_0}^1 : i_0 \in [n]\}$.
        Combining this with \Cref{lem:bound-k}, all jobs from $\mathcal{J}_j^*$ (and $\{F_{i_0}^1 : i_0 \in [n]\}$ if $j=1$) which are completed in the interval~$I$ are tardy and all jobs from $\mathcal{J}_j$ completed in the interval~$I$ are early.
        The only jobs from~$\mathcal{J}_j$ (and $\{F_{i_0}^1: i_0 \in [n]\}$) which can be completed in the interval~$I$ and be early are $J_{i_0, j}$ and $\neg J_{i_0, j}$ for $i_0 \ge i$.
        As we know that at least one of~$J_{i_0, j}$ and $\neg J_{i_0, j}$ is tardy for $i_0 \ge i$, at most $n + 1 -i$ jobs from $\mathcal{J}_j$ are scheduled in the interval~$I$.
        The only jobs from $\mathcal{J}^*_j$ which can be completed in the interval~$I$ are $J_{i_0, j}^*$ and $\neg J_{i_0, j}^*$ (and $F_{i_0}^1$ if $j=1$) for $i_0 \ge i$, and these jobs must be tardy (as the start point of~$I$ is later than the due date of every job from~$\mathcal J_j^*$).
        Note that exactly one of $J_{i_0, j}^*$ and $\neg J_{i_0, j}^*$ (and $F_{i_0}^1$ if $j= 1$) is tardy for $i_0 > i$ by the induction hypothesis.
        Thus, apart from $J^*_{i, j}$ and $\neg J_{i, j}^*$ (and $F_{i_0}^1$ if $j= 1$), at most $2n -2i + 1$ large jobs can be completed in the interval~$I$.
        We will show that at least $2n -2i + 2$ large jobs must be completed in the interval~$I$, implying that one of $J_{i,j}^*$ and $\neg J_{i, j}^*$ (and $F_{i_0}^1$ if $j= 1$) is tardy.
        Because $\alpha $ is chosen large enough so that all terms of order $\alpha^2$ or lower of all processing times together are smaller than $0.1 \cdot \alpha^3$, the start point of~$I$ is smaller than $(4nj - 2n + 2i -1) \cdot \alpha^3 $.
        Since the end point of the interval is larger than $(4nj + 0.1) \cdot \alpha ^3$ as $4nj \cdot \alpha^3 < \Delta_j$, \Cref{obs:rounding} implies that for each $x \in \{4nj -2n +  2i -1, 4nj-2i, \ldots, 4nj\}$, there is one large job completed in the interval~$[x \cdot \alpha^3, (x + 0.1) \cdot \alpha^3] \subseteq I$.
        Consequently, there are at least $2n -2i + 2$ large jobs which are completed within $I$.

        We continue with the third step, which is analogous to the second step.
        Consider the interval $I^* = [\Delta^*_j - (2n - 2i + 1.1) \cdot \alpha^3, \Delta^*_j + 0.1 \cdot \alpha^3]$.
        By \Cref{lem:possible-jobs-in-halfperiod}, the only large jobs which can be scheduled in~$I^*$ are $\mathcal{J}_j^* \cup \mathcal{J}_{j-1}$.
        Combining this with \Cref{lem:bound-k}, all jobs from $\mathcal{J}_{j-1}$ which are completed in~$I^*$ are tardy and all jobs from $\mathcal{J}^*_j$ completed in~$I^*$ are early.
        The only jobs from~$\mathcal{J}^*_j$ which can be completed in~$I^*$ and be early are $J^*_{i_0, j}$ and $\neg J^*_{i_0, j}$ for $i_0 \ge i$.
        As we know that at least one of $J^*_{i_0, j}$ and $\neg J^*_{i_0, j}$ is tardy for $i_0 \ge i$, at most $n + 1 -i$ jobs from $\mathcal{J}^*_j$ are scheduled in~$I^*$.
        The only jobs from $\mathcal{J}_{j-1}$ which can be completed in~$I^*$ are $J_{i_0, j-1}$ and $\neg J_{i_0, j-1}$ for $i_0 \ge i$, and these jobs must be tardy (as the start point of~$I^*$ is larger than the due date of every job from $\mathcal J_{j-1}$).
        Note that exactly one of $J_{i_0, j-1}$ and $\neg J_{i_0, j-1}$ is tardy for $i_0 > i$.
        Thus, apart from $J_{i, j-1}$ and $\neg J_{i, j-1}$, at most $2n -2i + 1$ large jobs can be completed in this interval.
        However, at least $2n -2i + 2$ large jobs must be completed in this interval, implying that one of $J_{i,j-1}$ and $\neg J_{i, j-1}$ is tardy:
        the start point of~$I^*$ is smaller than $(4n(j-1) + 2i -1) \cdot \alpha^3 $ since all lower-order terms together are smaller than $0.1 \cdot \alpha^3$.
        However, the end point of~$I^*$ is larger than $(4n(j-1) + 2n + 0.1) \cdot \alpha ^3$ as $(4n(j-1) + 2n ) < \Delta_j^*$.
        Consequently, \Cref{obs:rounding} implies that for each $ x \in \{4n(j-1) + 2i -1, \ldots, 4nj\}$, there is one large job completed in the interval $[x \cdot \alpha^3 , (x + 0.1) \cdot \alpha^3] \subseteq I^*$, implying that there are $2n - 2i + 2$ large jobs completed in~$I^*$.

        We continue with \Cref{item:at-most}, which is very similar to \Cref{item:at-least}.
        We again start with the first step.
        By \Cref{lem:bound-k}, there are $2n $ early jobs from $\{F_{i_0}^1 :i_0 \in [n]\} \cup \mathcal{J}_1^*$.
        Further, every large job completed until time~$\Delta_1^*$ is early.
        The only jobs from $\{F_{i_0}^1 :i_0 \in [n]\} \cup \mathcal{J}_1^*$ which can be early and completed after time~$T:= \Delta_1^* - (2n - 2i + 1.1) \cdot \alpha^3$ are $F_{i_0}^1$, $J^*_{i_0,1}$, and $\neg J^*_{i_0, 1}$ for $i_0 \ge i$.
        By the induction on~$i$, we know that exactly one of $F_{i_0}^1$, $J^*_{i_0, 1}$, and $\neg J^*_{i_0, 1}$ for $i_0 > i$ is tardy.
        As there are at least $2n - 2i +2 $ large jobs completed in the interval $[T, \Delta_1^*]$ (since at most $2i -2$ large jobs can be completed until time~$T < (2i-1) \cdot \alpha^3$), it follows that at least two jobs from~$F_i^1$, $J^*_{i, 1}$, and $\neg J^*_{i,1}$ are completed by time~$\Delta_1^*$ and thus early by \Cref{lem:bound-k}.
        Consequently, at most one of $F_i^1$, $J^*_{i,1}$, and $\neg J^*_{i,1}$ is tardy.
        
        We continue with the second step, showing that at most one of $J^*_{i, j}$ and $\neg J^*_{i, j}$ (and $F_{i}^1$ if $j=1$) being tardy implies that at most one of $J_{i,j}$ and $\neg J_{i,j}$ is tardy.
        Consider the interval $I = [\Delta_j - (2n - 2i + 1.1) \cdot \alpha^3, \Delta_j+0.1\cdot \alpha^3]$.
        By \Cref{lem:possible-jobs-in-halfperiod}, the only large jobs which can be scheduled in~$I$ are $\mathcal{J}_j^* \cup \mathcal{J}_j$ (and $\{F_{i_0}^1: i_0\in [n]\}$ if $j= 1$).
        Combining this with \Cref{lem:bound-k}, all jobs from $\mathcal{J}_j^*$ (and $\{F_{i_0}^1: i_0\in [n]\}$ if $j= 1$) which are completed in~$I$ are tardy and all jobs from $\mathcal{J}_j$ completed in~$I$ are early.
        The only jobs from $\mathcal{J}^*_j$ (and $\{F_{i_0}^1: i_0\in [n]\}$ if $j= 1$) which can be completed in~$I$ are $J_{i_0, j}^*$ and $\neg J_{i_0, j}^*$ (and $F_{i_0}^1$ if $j=1$) for $i_0 \ge n + 1 - i$, and these jobs must be tardy (as the start point of~$I$ is later than the due date of every job from $\mathcal J_j^*$).
        As we know that at most one of~$J_{i_0, j}$ and $\neg J_{i_0, j}$ is tardy for $i_0 \ge i$, at most $n + 1 -i$ jobs from $\mathcal{J}_j$ are scheduled in~$I$.
        The only jobs from~$\mathcal{J}_j$ (and $\{F_{i_0}^1: i_0\in [n]\}$ if $j= 1$) which can be completed in this interval and be early are $J_{i_0, j}$ and $\neg J_{i_0, j}$ (and $F_{i_0}^1$ if $j= 1$) for $i_0 \ge i$.
        As we know that at exactly of $J_{i_0, j}$ and $\neg J_{i_0, j}$ (and $F_{i_0}^1$ if $j= 1$) is tardy for $i_0 > i$, at most $n + 1 -i$ jobs from $\mathcal{J}_j$ (and $F_{i_0}^1$ if $j= 1$) are completed in~$I$.
        Thus, apart from $J_{i, j}$ and $\neg J_{i, j}$, at most $2n -2i + 1$ large jobs can be completed in~$I$.
        However, at least $2n -2i + 2$ large jobs must be completed in~$I$, implying that one of $J_{i,j}$ and $\neg J_{i, j}$ is tardy:
        the start point of~$I$ is smaller than $(4nj -2n +  2i -1) \cdot \alpha^3 $ while the end point of~$I$ is larger than $(4nj + 0.1) \cdot \alpha ^3$.
        By \Cref{obs:rounding}, for each $x \in \{4nj-2n + 2i-1, \ldots, 4nj\}$, there is a large job completed in~$[x \cdot \alpha^3, (x + 0.1) \cdot \alpha^3] \subseteq I$, implying that there are at least $2n -2i +2$ large jobs completed in~$I$.

        We conclude with the third step, which is again analogous to the second step.
        Consider the interval $[\Delta^*_{j + 1} - (2n - 2i + 1.1) \cdot \alpha^3, \Delta^*_{j + 1} + 0.1 \cdot \alpha^3]$.
        By \Cref{lem:possible-jobs-in-halfperiod}, the only large jobs which can be scheduled in this interval are $\mathcal{J}_{j + 1}^* \cup \mathcal{J}_{j}$.
        Combining this with \Cref{lem:bound-k}, all jobs from $\mathcal{J}_{j}$ which are completed in this interval are tardy and all jobs from $\mathcal{J}^*_{j + 1}$ completed in this interval are early.
        The only jobs from~$\mathcal{J}^*_{j +1}$ which can be completed in this interval and be early are $J^*_{i_0, j+1}$ and $\neg J^*_{i_0, j+1}$ for $i_0 \ge i$.
        As we know that at most one of $J^*_{i_0, j + 1}$ and $\neg J^*_{i_0, j + 1}$ is tardy for $i_0 \ge i$, at most $n + 1 -i$ jobs from $\mathcal{J}^*_j$ are scheduled in this interval.
        The only jobs from $\mathcal{J}_{j}$ which can be completed in this interval are $J_{i_0, j}$ and $\neg J_{i_0, j}$ for $i_0 \ge i$, and these jobs must be tardy (as the start point of the interval is later than the due date of every job from $\mathcal J_{j}$).
        Note that exactly one of $J_{i_0, j}$ and $\neg J_{i_0, j}$ is tardy for $i_0 > i$.
        Thus, apart from $J_{i, j}$ and $\neg J_{i, j}$, at most $2n -2i + 1$ large jobs can be completed in this interval.
        However, at least $2n -2i + 2$ large jobs must be completed in this interval, implying that one of $J_{i,j}$ and $\neg J_{i, j}$ is tardy:
        the start point of the interval is smaller than $(4nj+ 2i -1) \cdot \alpha^3 $ while the end point of the interval is larger than $(4nj + 2n + 0.1) \cdot \alpha ^3$.
        Thus, \Cref{obs:rounding} implies that for each $ x \in \{4nj + 2i-1, \ldots, 4nj\}$, there is one large job completed in the interval $[x \cdot \alpha^3, (x+0.1) \cdot \alpha^3]$.
    \end{proof}

    We can finally show \Cref{lem:only-candidate}:
    \begin{proof}[Proof of \Cref{lem:only-candidate}]
        Let $\sigma$ be an optimal schedule.
        By \Cref{obs:almost-edd}, we may assume that $\sigma $ is the canonical schedule for some set~$\widetilde{\mathcal{J}}$.
        By \Cref{lem:bound-k}, $\widetilde{\mathcal{J}}$ contains every delimiter job and job~$F_i^m$ for every $i \in [n]$.
        \Cref{lem:upper-lower} implies that $\mathcal{\widetilde{J}}$ contains exactly one of $J_{i,j}^*$ and $\neg J_{i,j}^*$ and exactly one of $J_{i,j}$ and $\neg J_{i,j}$ for every $i \in [n]$ and $j\in [m]$.
        Since $\sigma$ has at most $k$ tardy jobs and there are $k$ tardy number jobs, this implies that all filler jobs~$F_i^1$ for $i \in[n]$ are contained in $\mathcal{\widetilde{J}}$.
    \end{proof}

\section{Lexicographically first minimizing $T_{\max}$ and then $\sum U_j$}
In this section, we show that $1|T_{\max}\le \ell, \sum U_j \le k|$ can be reduced to the problem of first minimizing the maximum tardiness and then the number of tardy jobs:

\begin{theorem}\label{thm:lex:Tmax-Uj}
    $1||Lex(T_{\max}, \sum U_j)$ is strongly NP-complete.
\end{theorem}

\begin{proof}
    We reduce from $1|T_{\max}\le \ell, \sum U_j \le k|$ which is strongly NP-complete by \Cref{thm:strong}.
    Let $\mathcal I =(\mathcal J = \{J_1, \ldots, J_n\}, \ell, k)$ be an instance of $1|T_{\max}\le \ell, \sum U_j \le k|$.
    Let $P:= \sum_{i=1}^n p(J_i)$.
    We assume, without loss of generality, that $\ell < P$ and $d(J) \le P$ for every $J \in \mathcal J$, and that there is some schedule with maximum tardiness at most~$\ell$.
    We create a job~$J^*$ with $p(J^*) = P$ and $d(J^*) =2P -\ell$.
    We claim that $\mathcal I':= \mathcal{J} \cup \{J^*\}$ has a schedule minimizing the maximum tardiness with at most $k + 1 $ tardy jobs if and only if $\mathcal I $ is a yes-instance of $1|T_{\max} \le \ell, \sum U_j \le k|$.

    As a first step of showing the correctness of the reduction, we show that the minimal maximum tardiness for $\mathcal I'$ is precisely $\ell$.
    Let $\sigma'$ be any schedule for $\mathcal I'$.
    If $J^*$ is not the last job in $\sigma'$, then the last job~$J$ in $\sigma'$ is completed at time~$2P$ and thus has tardiness at least $2P - d(J) \ge P > \ell$.
    Otherwise, $J^*$ is the last job in~$\sigma'$ and has tardiness exactly $2P - d(J^*) = \ell$, implying that every schedule for $\mathcal I'$ has maximum tardiness at least $\ell$.
    To show that the minimum maximum tardiness is exactly $\ell$, it suffices to describe a schedule with maximum tardiness $\ell$.
    Note that taking the presumed schedule for $\mathcal{I}$ with maximum tardiness at most $\ell$ and appending $J^*$ is such a schedule.

    We conclude the proof by showing correctness.
    We start with the forward direction.
    Given a solution~$\sigma$ to $\mathcal I$, appending~$J^*$ to $\sigma$ results in a schedule for $\mathcal I'$ with one more tardy job (namely~$J^*$) and maximum tardiness $\ell$.

    Given a schedule $\sigma'$ for $\mathcal I'$ minimizing the maximum tardiness and having at most $k+1 $ tardy jobs, first note that any schedule minimizing $T_{\max}$ must schedule job~$J^*$ last.
    Consequently, job~$J^*$ is tardy.
    Further, the minimal maximum tardiness of a schedule is exactly~$\ell$, so job~$J$ has tardiness at most $\ell $ for every $J \in \mathcal J$.
    Thus, deleting $J^*$ from the schedule $\sigma $ leads to a schedule for $\mathcal{I}$ with maximum tardiness at most $\ell$ and at most $k$ tardy jobs.
\end{proof}

We will now show that the a priori version is strongly NP-complete:
\begin{corollary}
    For any $\alpha > 0$, $1||T_{\max} + \alpha \sum U_j$ is strongly NP-complete.
\end{corollary}

\begin{proof}
    Containment in NP is obvious as an optimal schedule is a certificate.
    We reduce from $1|| Lex(T_{\max}, \sum U_j)$.
    The reduction just multiplies each due date and each processing time by $2n \cdot \lceil \alpha\rceil$.
    This implies that the tardiness of each job is a multiple of $2n  \cdot \lceil {\alpha}\rceil > \alpha \sum U_j$.
    Consequently, any schedule minimizing $T_{\max}  + \alpha \sum U_j$ is a schedule that minimizes $T_{\max}$, and among all such schedules (i.e., those minimizing $T_{\max}$), the schedule with minimum number of tardy jobs minimizes $T_{\max} + \alpha \cdot \sum U_j$.
\end{proof}

\section{Lexicographically first minimizing $\sum U_j$ and then $T_{\max}$}

\newcommand{\Jweight}{X}
\newcommand{\tildeweight}{W}
\newcommand{\oldx}{J^*}
\newcommand{\oldy}{J}

We next prove the following theorem:
\begin{theorem}\label{thm:lex}
$1||Lex ( \sum U_j, T_{\max})$ is weakly NP-complete.
\end{theorem}

We will reduce from \textsc{Partition}:
\decprob{Partition}{
A set of $n$ integers~$a_1, \ldots, a_{n}$.
}{
Is there a subset~$S \subset [n]$ such that $\sum_{i \in S} a_i = t$ where $t := 0.5 \cdot \sum_{i=1}^n a_i$?
}
Somewhat similar to \Cref{sec:strong-NP}, there will be four jobs~$\oldx_i$, $\neg \oldx_i$, $\oldy_i$, and $\neg \oldy_i$ for each $i \in [n]$.
    Again, scheduling~$\oldx_i$ and $\oldy_i$ early will encode $a_i$ being part of the solution to \textsc{Partition} while $\neg \oldx_i$ and $\neg \oldy_i$ being early will encode $a_i$ not being part of the solution to \textsc{Partition}.
However, while in \Cref{sec:strong-NP} all these jobs had roughly the same processing time, this is not true anymore; in order to prevent that scheduling e.g.\ $\oldx_1$ at the end allows to schedule one more early jobs, jobs~$\oldy_i$ and $\neg \oldy_i$ will have much larger processing times than $\oldx_i$ and $\neg \oldx_i$, and the processing times will be increasing for increasing~$i$.
We add many small filler jobs~$F_i$ to ensure that $\oldx_i$ or $\neg \oldx_i$ together with the filler jobs have roughly the same processing time as $\oldy_i$ and $\neg \oldy_i$.
Since all these filler jobs have the same characteristics, we will refer to each of these jobs as~$F_i$; the set of all $2W/X$ filler jobs $F_i$ will be denoted as $\mathcal{F}_i$.
The processing times of the filler jobs will be quite small, ensuring that any schedule with minimum number of tardy jobs schedules them early.
We fix sufficiently large constants $W \gg X \gg Y\gg Z \gg t$ such that $W$ is a multiple of $X$ ($Z = (2t + 1)$, $Y = (2t + 1) \cdot Z$, $X = n \cdot 2^{n+2} \cdot Y$, and $W = 2n^2 \cdot X$ are possible choices).

For an overview of processing times and due dates, we refer to \Cref{tab:lex}.
We set the target maximum tardiness to~$\ell := n \cdot \tildeweight + \sum_{i=1}^{n} i\cdot \Jweight  + \sum_{i =1}^n 2^i \cdot Y + t \cdot Z + t$.
\begin{table}[h!]
    \begin{center}
    \begin{tabular}{c|c|c | c}
        Job & processing time & due date & multiplicity\\
        \hline
        $\oldx_i$ & $i \cdot \Jweight$ & $i \cdot \tildeweight + \sum_{i_0=1}^i i_0 \cdot \Jweight +  t$ & 1\\
        $\neg \oldx_i $ & $ i \cdot \Jweight + a_i$ & $(i -1) \cdot \tildeweight + \sum_{i_0=1}^i i_0 \cdot \Jweight + t$& 1\\
        \hline
        $F_i$ & $X/2 $ & $i \cdot \tildeweight + \sum_{i_0=1}^i i_0 \cdot \Jweight + t $& $2W/X$\\
        \hline
        \hline
        $\oldy_i$ & $\tildeweight + 2^i \cdot Y + a_i\cdot  Z$ & $D_1^* + i \cdot \tildeweight + \sum_{i_0=1}^{i} i_0 \cdot \Jweight + \sum_{i_0=1}^i 2^{i_0}  \cdot Y + t\cdot Z + t $& 1\\
        $\neg \oldy_i$& $\tildeweight + 2^i \cdot Y $ & $D_1^* + i \cdot \tildeweight + \sum_{i_0=1}^{i-1} i_0 \cdot \Jweight +  \sum_{i_0=1}^i 2^{i_0} \cdot Y + t \cdot Z+ t$& 1
    \end{tabular}
    \end{center}
    \caption{Processing times and due dates where $D_1^* := n \cdot W+ \sum_{i_0 =1}^n i_0 \cdot X +  t$.}
    \label{tab:lex}
\end{table}

We now prove correctness.
We start with characterizing the set of early jobs of a schedule minimizing the number of tardy jobs.
We start with the early jobs completed by time $D_1^*$:
    \begin{lemma}\label{claim:one-of-each}
        Any set~$\mathcal{J}^*$ of $n \cdot 2\tildeweight/\Jweight + i $ early jobs completed by time $D_1^*$ 
        has total processing time at least~$n \cdot \tildeweight + \sum_{{i_0} =1}^i {i_0} \cdot \Jweight$.
        If the total processing time of $\mathcal J^*$ is smaller than~$n \cdot \tildeweight + \sum_{{i_0} =1}^i {i_0} \cdot\Jweight + 0.5 \cdot \Jweight$, then $\mathcal{J}^*$ contains exactly one of $\oldx_{i_0}$ and $\neg \oldx_{i_0}$ for each ${i_0} \in [i]$ as well as each filler job~$F_{i_0}$.
    \end{lemma}

    \begin{proof}
        We first show that there is no $i \in [n]$ such that $\mathcal{J}^*$ contains $J_i $ or $\neg J_i$.
        Each job has processing time at least $X/2$.
        Thus, if $\mathcal{J}^*$ would contain job $J_i $ or $\neg J_i$ which has processing time at least $W$, then
        \[
         p(\mathcal{J}^*) \ge (n \cdot 2W/X) \cdot X/2 + W = (n + 1) \cdot W>D_1^*\,,
        \]
        a contradiction to all jobs from $\mathcal{J}^*$ finishing until $D_1^*$.
    
        We prove the statement by induction on~$i$.
        For $i =1$, note that each non-filler job has processing time at least~$\Jweight$ and only $\oldx_1$ and $\neg \oldx_1$ have a processing time smaller than $2 \cdot \Jweight$ while each filler job has processing time~$\Jweight/2$.
        Thus, each set of $n \cdot 2\tildeweight/\Jweight + 1$ jobs has processing time at least~$n\cdot \tildeweight + \Jweight$.
        Further, each set of $n \cdot 2\tildeweight /\Jweight + 1$ early jobs containing at least two non-filler has processing time at least 
        \[
            (n \cdot 2\tildeweight/\Jweight - 1) \cdot (X/2) + 2\cdot X = n \cdot W + 1.5\cdot X,
        \]
        and each set of $n  \cdot 2\tildeweight /\Jweight + 1$ early jobs containing at least one non-filler job apart from $J_1^*$ and $\neg J_1^*$ has processing time at least
        \[
            (n \cdot 2\tildeweight/\Jweight) \cdot (X/2) + 2\cdot X = n \cdot W + 2\cdot X.
        \]

        We continue with the induction step.
        Let $i>1$ and assume that the lemma holds for $i-1$.
        Let~$ J^* $ be the longest job from~$\mathcal{J}^*$.
        As there are only $n \cdot 2W/X$ filler jobs and $\mathcal{J}^*$ contains no job from $\{J_{i_0}, \neg J_{i_0} :i_0 \in [n]\}$ by our initial observation, $|\mathcal J^*| > n \cdot 2W/X$ implies that $J^* = \oldx_{i_1}$ or $J^* = \neg \oldx_{i_1}$ for some~$i_1 \in [n]$.
        We will now show that $i_1 \ge i$, i.e., that the longest job from $\mathcal{J}^*$ has index at least $i$.
        So assume towards a contradiction that $i_1 < i$,  i.e., all non-filler jobs from $\mathcal{J}^*$ have index smaller than $i$.
        Since $|\mathcal J^*| =  n \cdot 2W /X + i >  n \cdot 2W /X + i_1$, this implies that there is some minimal~$i_2$ such that $|\bigl(\{\oldx_{i_0}, \neg \oldx_{i_0} : i_0 \le i_2\} \cup  \mathcal F_{1} \cup \ldots \cup \mathcal F_{i_2} \bigr) \cap \mathcal J^*| > i_2 \cdot (2W/X  + 1)$.
        By the minimality of~$i_2$, we have $\{\oldx_{i_2}, \neg \oldx_{i_2}\}  \cup \mathcal F_{i_2}\subset \mathcal J^*$.
        Further, due to the minimality of~$i_2$, we have that $p\Bigl(\bigl(\{\oldx_{i_0}, \neg \oldx_{i_0}: i_0 < i_2\} \cup  \mathcal F_{1} \cup \ldots \cup \mathcal F_{i_2-1}\bigr)  \cap \mathcal J^*\Bigr) \ge (i_2 - 1) \cdot W + \sum_{i_0 = 1}^{i_2-1} i_0 \cdot X$.
        Consequently, we have 
        \begin{align*}
        p\Bigl(\bigl(\{\oldx_{i_0}, \neg \oldx_{i_0}: i_0 \le i_2\} \cup  \mathcal F_{1} \cup \ldots \cup \mathcal F_{i_2}\bigr)  \cap \mathcal J^*\Bigr) & > (i_2 -1 ) \cdot W + \sum_{i_0 = 1}^{i_2-1} i_0 \cdot X+ W + 2\cdot i_2 \cdot X \\
        &= i_2 \cdot W + \sum_{i_0 = 1}^{i_2} i_0 \cdot X + i_2 \cdot X > d(F_{i_2})\,,
        \end{align*}
        a contradiction to all jobs from~$\mathcal J^*$ being early.
        Consequently, we have $p(J^*) \ge i \cdot X$.
        By induction, $\mathcal{J}^* \setminus \{J^*\}$ has a total processing time of at least $n\cdot \tildeweight + \sum_{{i_0}=1}^{i-1} {i_0} \cdot \Jweight$.
        Therefore, we have
        \[
            p(\mathcal J^*) = p(\mathcal J^*\setminus \{J^*\}) + p( J^*) \ge n \cdot W + \sum_{{i_0}=1}^{i-1} {i_0} \cdot \Jweight + i \cdot \Jweight = n \cdot W + \sum_{{i_0} =1}^i {i_0} \cdot \Jweight\,
        \]
        and if $p(\mathcal J^*) < n \cdot W + \sum_{{i_0} =1}^{i} {i_0} \cdot \Jweight +0.5 \cdot \Jweight$, then $\mathcal J^*$ contains exactly one of $\oldx_{i_0}$ and $\neg \oldx_{i_0}$ for each ${i_0} \in [i]$.
    \end{proof}

Using \Cref{claim:one-of-each}, we now characterize the set of schedules with minimum number of early jobs (ignoring the secondary condition of minimizing the maximum tardiness for a moment).
In order to do so, we introduce {EDD-schedules}, a class of schedules used to solve $1||\sum w_j U_j$.
An \emph{EDD-schedule for a set $\widetilde{\mathcal{J}}$} is a schedule which schedules all jobs from $\widetilde{\mathcal{J}}$ by non-decreasing order of their due dates, and then all other jobs in arbitrary order afterwards.
It is a well-known fact that if there is a schedule where a set $\widetilde{\mathcal{J}}$ of jobs is early, then the EDD-schedule for~$\widetilde{\mathcal{J}}$ is such a schedule (see e.g.~\cite{AdamuAdewumi2014}).
Note that if $\ell$ exceeds the total processing time of all jobs together, then the canonical schedule for $\widetilde{\mathcal{J}}$ (as defined in \Cref{sec:prelims}) is an EDD-schedule for~$\widetilde{\mathcal{J}}$.

\begin{lemma}\label{lem:structure-min-tardy-jobs}
    A schedule minimizes the number of tardy jobs if and only if the set~$\widetilde{\mathcal{J}}$ of early jobs consists of
    \begin{itemize}
        \item all filler jobs,
        \item for every $i\in[n]$, exactly one of $\oldx_i$ and $\neg \oldx_i$, 
        \item $\sum_{i \in [n] : \neg \oldx_i \in \widetilde{\mathcal{J}}} a_i  \le t$, and
        \item exactly $n$ jobs from $\{\oldy_i, \neg \oldy_i : i \in [n]\}$ such that for each $i \in [n]$, at most $i $ jobs from $\{\oldy_{i_0}, \neg \oldy_{i_0}: i_0 \in [i]\}$ are early.
    \end{itemize}
\end{lemma}

\begin{proof}
    We first show that for set $\mathcal{\widetilde{J}}$ fulfilling the four bullet points, the EDD-schedule~$\sigma$ for $\widetilde{\mathcal{J}}$ schedules every job from $\mathcal{\widetilde{J}}$ early.
    We start with the filler jobs and jobs~$J_i^*$.
    Until~$F_i$ or $J_i^*$, exactly one of $\oldx_{i_0}$ and $\neg \oldx_{i_0}$ for $i_0 \le i$ and $\mathcal{F}_{i_0}$ for $i_0\le i$ can be scheduled.
    Thus, $F_i$ or $J_i^*\in \mathcal{\widetilde{J}}$ is completed by time
    \begin{align*}
        \sum_{i_0 \in[i] : \oldx_{i_0} \in \mathcal{\widetilde{J}}} p(J_{i_0}^*) +\sum_{i_0 \in[i] : \neg \oldx_{i_0} \in \mathcal{\widetilde{J}}} p(\neg J_{i_0}^*) + \sum_{i_0 =1}^{i}p(\mathcal{F}_{i_0})& = \sum_{i_0=1}^i i_0 \cdot X + \sum_{i_0\in[i] : \neg \oldx_{i_0} \in \mathcal{\widetilde{J}}} a_{i_0} + i \cdot W \\
        & \le i \cdot W + \sum_{i_0 = 1}^i i_0 \cdot X + t = d(F_i) = d(J_i^*)
    \end{align*}
    using the third bullet point for the inequality.

    We continue with job $\neg J_i^*\in \mathcal{\widetilde{J}}$.
    Until $\neg J_i^*$, exactly one of $\oldx_{i_0}$ and $\neg \oldx_{i_0}$ for $i_0 \le i$ and $\mathcal{F}_{i_0}$ for $i_0 < i$ is scheduled.
    Thus, $\neg J_i^*\in \mathcal{\widetilde{J}}$ is completed by time
    \begin{align*}
        \sum_{i_0 \in[i] : \oldx_{i_0} \in \mathcal{\widetilde{J}}} p(J_{i_0}^*) +\sum_{i_0 \in[i] : \neg \oldx_{i_0} \in \mathcal{\widetilde{J}}} p(\neg J_{i_0}^*) + \sum_{i_0 =1}^{i-1}p(\mathcal{F}_{i_0}) & = \sum_{i_0=1}^i i_0 \cdot X + \sum_{i_0\in[i] : \neg \oldx_{i_0} \in \mathcal{\widetilde{J}}} a_{i_0} + (i -1)\cdot W \\
        & \le (i-1) \cdot W + \sum_{i_0 = 1}^i i_0 \cdot X + t = d(\neg J_i^*)
    \end{align*}
    using again the third bullet point for the inequality.

    Finally, we consider job $J_i\in \mathcal{\widetilde{J}}$ or $\neg J_i\in \mathcal{\widetilde{J}}$.
    Until $J_i$ or $\neg J_i$, all filler jobs, exactly one of $\oldx_{i_0} $ and $\neg \oldx_{i_0}$ for each $i_0 \in [n]$, and up to $i$ jobs from $\{\oldy_{i_0}, \neg \oldy_{i_0}: i_0 \in [i]\}$ (by the fourth bullet point) are scheduled.
    Thus, the job is completed by time
    \begin{align*}
        \sum_{i_0=1}^n \bigl(p(\mathcal{F}_{i_0} ) + p(\neg J_i^*) \bigr) + \sum_{J \in \{\oldy_{i_0}, \neg \oldy_{i_0}: i_0 \in [i]\} \cap \widetilde{\mathcal{J}}} p(J) & = n \cdot W + \sum_{i_0=1}^n i_0 \cdot  X+ 2t +  \sum_{J \in \{\oldy_{i_0}, \neg \oldy_{i_0}: i_0 \in [i]\} \cap \widetilde{\mathcal{J}}} p(J)\\
        & \le D_1^* + t + i \cdot W + i \cdot 2^{i} \cdot Y  + \sum_{i_0 :\oldy_{i_0}\in \widetilde{\mathcal{J}}} a_{i_0} \cdot Z\\
        & < D_1^* + i \cdot W + X\,.
    \end{align*}
    Note that this last term is smaller than $d(J_i)$ and $d(\neg J_i)$ except for $\neg J_1$.
    For $\neg J_1$, note that $\{\oldy_{i_0}: i_0 \in [i]\} \cap \widetilde{\mathcal{J}} = \emptyset$ and thus, $\neg J_1$ is completed by time $D_1^* + t + W + 2 \cdot Y < d(\neg J_1)$.

    We continue with the backwards direction.
    We show that any schedule has at most $n \cdot (2W/X) + 2n$ early jobs, and equality only holds if the four bullet points hold.
    By \Cref{claim:one-of-each}, there are at most $ 2 n \cdot W/X +n$ early jobs completed by time~$D_1^*$.
    We continue by analyzing the early jobs completed after $D_1^*$.
    The only jobs with due date larger than $D_1^*$ are $\{\oldy_i, \neg \oldy_i : i\in [n]\}$, all of which have processing time at least~$W$.
    Since the largest due date is $d(\oldy_n) < D_1^* + (n+1) \cdot W$, it follows that at most $n$ jobs completed after~$D_1^*$ are early.
    As there is a schedule with $2n \cdot W/X + 2n $ early jobs (as shown in the forward direction), this implies that the fourth bullet point holds.

    We continue by showing the first two bullet points.
    Note that $D_1^* < n \cdot W +   \sum_{i_0=1}^n i_0\cdot X + 0.5 \cdot X$.
    \Cref{claim:one-of-each} then implies that all filler jobs and exactly one of $\oldx_i$ and $\neg \oldx_i$ for $i \in [n]$ is early.

    It remains to show the third bullet point.
    We may assume, without loss of generality and due to the EDD order, that the last early job from $\{J_i^*, \neg J_i^* :i \in [n]\} \cup \mathcal{F}_1 \cup \ldots \cup \mathcal F_n$ is~$F_n$.
    Then $F_n$ is completed by time
    \begin{align*}
        \sum_{i_0 \in[n] : \oldx_{i_0} \in \mathcal{\widetilde{J}}} p(J_{i_0}^*) +\sum_{i_0 \in[n] : \neg \oldx_{i_0} \in \mathcal{\widetilde{J}}} p(\neg J_{i_0}^*) + \sum_{i_0 =1}^{n}p(\mathcal{F}_{i_0})& = \sum_{i_0=1}^n i_0 \cdot X + \sum_{i_0\in[n] : \neg \oldx_{i_0} \in \mathcal{\widetilde{J}}} a_{i_0} + n \cdot W \\
        & = d(F_n) + \sum_{i_0\in[n] : \neg \oldx_{i_0} \in \mathcal{\widetilde{J}}} a_{i_0} - t\,.
    \end{align*}
    Because $F_n$ is early, this implies the third bullet point.
\end{proof}

Since \Cref{lem:structure-min-tardy-jobs} implies that the minimum number of tardy jobs is $2n$ in any schedule, we will use $k:=2n$ in the following.
If we require that the maximum tardiness is at most $\ell$ in addition to the schedule having only $k$ tardy jobs, then we may restrict ourselves to canonical schedules by \Cref{obs:almost-edd} and we also get structural restrictions on which jobs from $\{\oldy_i, \neg \oldy_i : i \in [n]\}$ are early:
\begin{lemma}\label{lem:y}
    Any canonical schedule with minimum number of tardy jobs and tardiness at most $\ell$ schedules exactly one of $\oldy_i $ and $\neg \oldy_i$ early for each $i \in [n]$.
\end{lemma}

\begin{proof}
    By \Cref{lem:structure-min-tardy-jobs}, there are $n$ early jobs from $\{\oldy_i, \neg \oldy_i : i \in[n]\}$.
    Assume towards a contradiction that the statement is not true, i.e., there exists a canonical schedule $\sigma $ and some~$i\in[n]$ such that both $\oldy_i$ and~$\neg \oldy_i$ are early (note that neither $\oldy_i$ nor~$\neg \oldy_i$ being early for some~$i \in [n]$ implies that there is some $i_0 \in [n]$ with~$\oldy_{i_0}$ and $\neg \oldy_{i_0}$ being early as there are $n$ early jobs from $\{J_i, \neg J_i : i\in [n]\}$ by \Cref{lem:structure-min-tardy-jobs}).
    First consider the case that there exist some $i$ so that there are at least $i+1$ early jobs from $\{\oldy_{i_0}, \neg \oldy_{i_0} : i_0 < i\}$.
    By \Cref{lem:structure-min-tardy-jobs}, all filler jobs are early.
    Since $d(F_{i_0}) < d(\neg J_i) < d(J_i)$ for every $i_0 \in[n]$, every filler job $F_{i_0}$ is completed before~$\neg J_i$ and $J_i$ since $\sigma $ is canonical.
    Consequently, $\oldy_i$ is completed not before time
    \[
        n \cdot W + (i-1) \cdot W + W + W = (n+i+1) \cdot W > d(\oldy_i)\,,
    \]
    a contradiction to $\oldy_i$ being early.

    We now consider the other case, i.e., that for every $i \in [n] $  there are at most $i$ early jobs from $\{\oldy_{i_0}, \neg \oldy_{i_0} : i_0 \le i\}$.
    Let $i$ be maximal so that $\oldy_i$ and $\neg \oldy_i$ are early (such an $i$ exists by our initial assumption).
    As there are exactly $n$ early jobs from $\{\oldy_{i_0}, \neg \oldy_{i_0} : i_0 \in [n]\}$, this implies that for all $i_1 > i$, exactly one of $\oldy_{i_1} $ and~$\neg \oldy_{i_1}$ is early.
    In particular, at least one of $\oldy_n$ and $\neg \oldy_n$ is early (both are early if $i=n$).
    If $\oldy_n $ is early, then when $\oldy_n $ is completed, the following jobs are completed as well:
all filler jobs,
jobs $\oldx_{i_0}$ and $\neg \oldx_{i_0} $ for $i_0 \in [n]$ (because they have tardiness at most $\ell$ by assumption and $d(\neg \oldx_{i_0}) + \ell < d(\oldx_{i_0}) + \ell < d(J_i)$),
$i-2$ jobs from $\{\oldy_{i_0}, \neg \oldy_{i_0} : i_0 < i\}$,
$\oldy_{i}$ and $\neg \oldy_i$, and
exactly one of $\oldy_{i_0}$ and $\neg \oldy_{i_0}$ for $i_0 > i$.
    Therefore, $\oldy_n$ is completed not before time
    \[
        n \cdot W + 2 \cdot \sum_{i_0=1}^n i_0 \cdot X + t + n \cdot W + 2 \cdot 2^i \cdot Y + \sum_{i_0 =i+1}^n 2^{i_0} \cdot Y > d(\oldy_n)\,,
    \]
    a contradiction to $\oldy_n$ being early.
    If $\oldy_n$ is tardy, then $i< n$ and $\neg \oldy_n$ is early.
    The observation that $d(\neg \oldy_n) < d(\oldy_n)$ together with the above calculation leads to a contradiction.
\end{proof}

We call a set $\widetilde{\mathcal J}$ fulfilling the conditions of \Cref{lem:structure-min-tardy-jobs,lem:y} (that is, containing every filler job and for each $i \in [n]$, exactly one of $\oldx_i$ and $\neg \oldx_i$ as well as exactly one of $\oldy_i$ and $\neg \oldy_i$, and fulfilling $\sum_{i: \neg J_i^*\in \mathcal{\widetilde{J}}} a_i \le t$) a \emph{candidate set}.
By \Cref{lem:structure-min-tardy-jobs}, we may restrict ourselves to schedules~$\sigma$ constructed by applying \Cref{obs:almost-edd} to a candidate set.
We will further assume that ties between $\oldx_i$ and $F_i$ will be broken in favor of $\oldx_i$ (i.e., if $\oldx_i \in \mathcal{\widetilde{J}}$ (note that $F_i$ is always in $\mathcal{\widetilde{J}}$), then we schedule $\oldx_i$ before $F_i$), and a possible tie between $\oldx_n$ and $\oldy_n$ will be broken in favor of $\oldy_n$.
Thus, given a candidate set~$\widetilde{\mathcal J}$, the corresponding schedule looks as follows:
\begin{enumerate}
    \item First, $i=1$ to $n$, job~$\oldx_i$ if $\oldx_i \in \widetilde{\mathcal J}$ or $\neg \oldx_i$ otherwise (i.e., $\neg \oldx_i \in \widetilde{\mathcal J}$). In both cases, $\mathcal F_i$ follows.
    \item Second, $i=1$ to $n$, the following jobs are scheduled:
    If $\oldx_i \in \widetilde{\mathcal J}$, then $\neg \oldx_i$ is scheduled, followed by~$\oldy_i$ if $\oldy_i \in \widetilde{\mathcal J}$ or $\neg \oldy_i$ if $\neg \oldy_i \in \widetilde{\mathcal J}$.
    Otherwise (i.e., if $\oldx_i \notin \widetilde{\mathcal J}$), we first schedule $\oldy_i$ if $\oldy_i \in \widetilde{\mathcal J}$ or $\neg \oldy_i$ if $\neg \oldy_i \in \widetilde{\mathcal J}$ and afterwards schedule~$\oldx_i$.
    \item Third, for $i=1$ to $n$, we schedule $\oldy_i $ if $\neg \oldy_i \in \widetilde{\mathcal J}$ or $\neg  \oldy_i$ if $\oldy_i \in \widetilde{\mathcal J}$.
\end{enumerate}
An example of the schedule for a candidate set can be found in \Cref{ex:lex}.
\begin{example}\label{ex:lex}
    Assume that $n =4$ and $\widetilde{\mathcal{J}} = \{\oldx_1, \neg \oldx_2, \oldx_3, \neg \oldx_4\} \cup \{\oldy_1, \neg \oldy_2 , \oldy_3, \neg \oldy_4\} \cup \{\mathcal F_1, \mathcal F_2, \mathcal F_3, \mathcal F_4\}$.
    Then the schedule for $\widetilde{\mathcal J }$ is
    \[
        \begin{array}{rrrrrrrr}
            \oldx_1, & \mathcal F_1, & \neg \oldx_2, & \mathcal F_2, &\oldx_3, & \mathcal F_3 ,& \neg \oldx_4, & \mathcal F_4, \\
            \neg \oldx_1, & \oldy_1, & \neg \oldy_2, & \oldx_2, & \neg \oldx_3 ,& \oldy_3 ,& \neg \oldy_4, & \oldx_4,\\
            \neg \oldy_1, & \oldy_2, & \neg \oldy_3, & \oldy_4.
        \end{array}
    \]
\end{example}
For a candidate set~$\widetilde{\mathcal J}$, we define $\widetilde S^* := \{i \in [n]: \oldx_i\in \widetilde{ \mathcal J}\}$ and $\widetilde S := \{i \in [n] : \oldy_i \in \widetilde{\mathcal J}\}$.

We now want to characterize when a candidate set~$\mathcal{\widetilde{J}}$ corresponds to a solution to the constructed instance of $1||Lex(\sum U_j, T_{\max})$.
We first show that all filler jobs as well as jobs $J_i^*$ or $\neg J_i^*$ from $\mathcal{\widetilde{J}}$ are early:

\begin{lemma}\label{lem:weak-early-x}
    Let $\widetilde{\mathcal J}$ be a candidate set.
    Then all jobs from $J \in \mathcal J^* := \widetilde{\mathcal J}  \cap \bigl(\{\oldx_i,\neg \oldx_i\} \cup \mathcal F_i\bigr)$ are early in the canonical schedule for $\widetilde{\mathcal J}$.
\end{lemma}
\begin{proof}
    Note that $\ell > D_1^* \ge d(J) $ for every $J \in \mathcal J^* := \widetilde{\mathcal J}  \cap \bigl(\{\oldx_i,\neg \oldx_i\} \cup \mathcal F_i\bigr)$.
    Therefore, these jobs are the first jobs in the canonical schedule~$\sigma$ for $\mathcal{\widetilde{J}}$, implying that they are early in $\sigma$ if and only if they are early in the EDD-schedule for $\mathcal{\widetilde{J}}$ which holds by \Cref{lem:structure-min-tardy-jobs}.
\end{proof}

Next, we analyze when the remaining jobs are early, i.e., the jobs from $\widetilde{\mathcal J} \cap \{\oldy_i, \neg \oldy_i: i\in[n]\}$.
\begin{lemma}\label{lem:weak-early-y}
    Let $\widetilde{\mathcal J}$ be a candidate set.
    Then all $n$ jobs from $\widetilde{\mathcal J} \cap \{\oldy_i, \neg \oldy_i : i\in[n]\}$ are early in the canonical schedule for $\widetilde{\mathcal J}$ if and only if 
    \begin{itemize}
        \item for each $i \in [n]$ with $\oldx_i \in\widetilde{\mathcal J }$, we have $\oldy_i\in \widetilde{\mathcal J}$ and $\sum_{i_0 \in [i] \cap \widetilde S} a_{i_0} \le t$, and 
        \item for each $i \in [n] $ such that $\neg J_i, \neg J_i^* \in \mathcal{\widetilde{J}}$, we have $\sum_{i_0 \in [i] \cap \widetilde S} a_{i_0} \le t$.
    \end{itemize}
\end{lemma}

\begin{proof}
    First consider the case that $\oldx_i \in \widetilde{\mathcal J}$.
    By the definition of a candidate set, exactly one of $\oldy_i$ and $\neg \oldy_i$ is contained in $\mathcal{\widetilde{J}}$.
    We will now show that job~$\neg \oldy_i$ will be tardy even if $\neg \oldy_i \in \mathcal{\widetilde{J}}$ while $\oldy_i$ is tardy if and only if $\oldy_i \notin \mathcal{\widetilde{J}}$ or $\sum_{i_0 \in [i] \cap \widetilde S} a_{i_0} > t$.
    Job $\oldy_i $ if $\oldy_i \in \widetilde{\mathcal J}$ or $\neg \oldy_i$ if $\neg \oldy_i \in \widetilde{\mathcal J}$ is completed after
all filler jobs,
jobs $\oldx_{i_0}$ and $\neg \oldx_{i_0} $ for $i_0 \in [i]$,
exactly one of $\oldx_{i_0}$ and $\neg \oldx_{i_0}$ for $i_0 > i$, 
job~$\oldy_{i_0}$ for $i_0 \in [i] \cap \widetilde S$, and
$\neg \oldy_{i_0}$ for $i_0 \in [i]\setminus \widetilde S$.
   To simplify the following equations, we assume that for each $i_0 > i$, job $\neg \oldx_{i_0}$ and not $ \oldx_{i_0}$ is scheduled before~$\oldy_i$ (this increases the total completion time by at most $2t< Z$).
    Thus, $\oldy_i$ or $\neg \oldy_i$ is completed by time
    \begin{align*}
        \sum_{i_0=1}^n p(\mathcal F_{i_0}) + & \sum_{i_0 =1}^{i} \Bigl( p(\oldx_{i_0}) + p(\neg \oldx_{i_0}) \Bigr) + \sum_{i_0 = i+1}^n p( \neg \oldx_{i_0}) + \sum_{i_0 \in [i] \cap \widetilde S} p(\oldy_{i_0}) + \sum_{i_0 \in [i]\setminus \widetilde S} p(\neg \oldy_{i_0}) \\
        & = n \cdot W + \sum_{i_0=1}^n i_0 \cdot X + \sum_{i_0=1}^i i_0 \cdot X + 2t + i \cdot W + \sum_{i_0 =1}^i 2^{i_0 } \cdot Y + \sum_{i_0 \in [i] \cap \widetilde S} a_{i_0} \cdot Z\\
        & = (n+i) \cdot W +  \sum_{i_0=1}^n i_0 \cdot X + \sum_{i_0=1}^i i_0 \cdot X +\sum_{i_0 =1}^i 2^{i_0 } \cdot Y + \sum_{i_0 \in [i] \cap \widetilde S} a_{i_0} \cdot Z  + 2t\\
        & = D_1^* + i \cdot W + \sum_{i_0=1}^i i_0 \cdot X + \sum_{i_0=1}^i 2^{i_0} \cdot Y  + \sum_{i_0 \in [i] \cap \widetilde S} a_{i_0} \cdot Z + t\\
        & = d(\oldy_i) + \bigl(\sum_{i_0 \in [i] \cap \widetilde S} a_{i_0} - t\bigr) \cdot Z \\
        & = d(\neg \oldy_i) + i \cdot X + \bigl(\sum_{i_0 \in [i] \cap \widetilde S} a_{i_0} - t\bigr) \cdot Z \,.
    \end{align*}
    Thus, if $\neg \oldy_i \in \widetilde{\mathcal J}$, then $\neg \oldy_i$ is tardy.
    If $\oldy_i \in \widetilde{\mathcal J}$, then $\oldy_i$ is early if and only if $\sum_{i_0 \in [i] \cap \widetilde S} a_{i_0} \le t$.

    We continue with the case $\neg \oldx_i \in \widetilde{\mathcal J}$.
    We will show that $\oldy_i $ is early if and only if $\oldy_i \in \mathcal{\widetilde{J}}$ while $\neg \oldy_i$ is early if and only if $\neg \oldy_i \in \mathcal{\widetilde{J}}$ and $\sum_{i_0 \in [i] \cap \widetilde S} a_{i_0} \le t$.
    Job $\oldy_i $ if $\oldy_i \in \widetilde{\mathcal J}$ or $\neg \oldy_i$ if $\neg \oldy_i \in \widetilde{\mathcal J}$ is completed after
all filler jobs,
jobs $\oldx_{i_0}$ and $\neg \oldx_{i_0} $ for $i_0 \in [i-1]$,
exactly one of $\oldx_{i_0}$ and $\neg \oldx_{i_0}$ for $i_0 \ge i$, 
$\oldy_{i_0}$ for $i \in [i] \cap \widetilde S$, and
$\neg \oldy_{i_0}$ for~$i \in [i]\setminus \widetilde S$.
    As in the case $\oldx_i \in \widetilde{\mathcal J}$, we assume that for each $i_0 \ge i$, job $\neg \oldx_{i_0}$ and not $ \oldx_{i_0}$ is scheduled before $\oldy_i$ (this increases the total completion time by at most $2t < Z$).
    Thus, $\oldy_i$ or $\neg \oldy_i$ is completed by time
    \begin{align*}
        \sum_{i_0=1}^n p(\mathcal F_{i_0}) + & \sum_{i_0 =1}^{i-1} \Bigl( p(\oldx_{i_0}) + p(\neg \oldx_{i_0}) \Bigr) + \sum_{i_0 = i}^n p(\neg \oldx_{i_0}) + \sum_{i_0 \in [i] \cap \widetilde S} p(\oldy_{i_0}) + \sum_{i_0 \in [i]\setminus \widetilde S} p(\neg \oldy_{i_0}) \\
        & = (n+i) \cdot W +  \sum_{i_0=1}^n i_0 \cdot X + \sum_{i_0=1}^{i-1} i_0 \cdot X +2t + \sum_{i_0 =1}^i 2^{i_0 } \cdot Y + \sum_{i_0 \in [i] \cap \widetilde S} a_{i_0} \cdot Z\\
        & = D_1^* + i \cdot W + \sum_{i_0=1}^{i-1} i_0 \cdot X + \sum_{i_0=1}^i 2^{i_0} \cdot Y  + \sum_{i_0 \in [i] \cap \widetilde S} a_{i_0} \cdot Z + t\\
        & = d(\oldy_i) - i \cdot X+ \bigl(\sum_{i_0 \in [i] \cap \widetilde S} a_{i_0} - t\bigr) \cdot Z \\
        & = d(\neg \oldy_i) + \bigl(\sum_{i_0 \in [i] \cap \widetilde S} a_{i_0} - t\bigr) \cdot Z
    \end{align*}
    Therefore, $\oldy_i$ is always early if $J_i \in \widetilde{\mathcal{J}}$ while $\neg \oldy_i$ is early if and only if $\neg \oldy_i \in \mathcal{\widetilde{J}}$ and $\sum_{i_0 \in [i] \cap \widetilde S} a_{i_0} \le t$.

    Since a candidate job contains exactly one of $\oldy_i$ and $\neg \oldy_i$ for every $i\in[n]$, $n$ jobs from $\mathcal{\widetilde{J}}$ are early if and only if exactly one of $\oldy_i$ and $\neg \oldy_i$ is early for every $i\in[n]$.
    As shown above, this is in turn equivalent to the two bullet points.
\end{proof}

Finally, we analyze when the tardiness is at most $\ell$.

\begin{lemma}\label{lem:weak-tardy}
    Let $\widetilde{\mathcal J }$ be a candidate set.
    Then the maximum tardiness is at most $\ell$ if and only if $\oldx_n \in \widetilde{\mathcal J}$ or $\sum_{i \in \widetilde S} a_i \le t$.
\end{lemma}

\begin{proof}
    We make a distinction between the different jobs.
    First consider a job~$\oldx_i$, and assume that $\oldx_i \notin \widetilde{\mathcal J}$ (otherwise $\oldx_i$ will be completed earlier than the time we compute below).
    Until $\oldx_i$, the following jobs are scheduled:
all filler jobs,
jobs $\oldx_{i_0}$ and $\neg \oldx_{i_0}$ for $i_0 \le i$,
exactly one of $\oldx_{i_0} $ and $\neg \oldx_{i_0}$ for $i_0  > i$, and
$\oldy_{i_0} $ for~$i_0 \in [i] \cap \widetilde S$, and
$\neg \oldy_{i_0}$ for $i_0 \in [i]\setminus \widetilde S$.
    Consequently, using that $p(J_{i_0}^*) \le p(\neg J_{i_0}^*)$, job~$\oldx_i$ is completed by time
    \begin{align*}
        \sum_{i_0=1}^n p(\mathcal F_{i_0}) + & \sum_{i_0 =1}^{i} \Bigl( p(\oldx_{i_0}) + p(\neg \oldx_{i_0}) \Bigr) + \sum_{i_0 = i+1}^n p(\neg \oldx_{i_0}) + \sum_{i_0 \in [i] \cap \widetilde S} p(\oldy_{i_0}) + \sum_{i_0 \in [i]\setminus \widetilde S} p(\neg \oldy_{i_0}) \\
        & = n \cdot W + \sum_{i_0=1}^n i_0 \cdot X + \sum_{i_0=1}^i i_0 \cdot X + \sum_{i_0=1}^n a_{i_0} + i \cdot W + \sum_{i_0 =1}^i 2^{i_0 } \cdot Y + \sum_{i_0 \in [i] \cap \widetilde S} a_{i_0} \cdot Z\\
        & = (n+i) \cdot W +  \sum_{i_0=1}^n i_0 \cdot X + \sum_{i_0=1}^i i_0 \cdot X +\sum_{i_0 =1}^i 2^{i_0 } \cdot Y + \sum_{i_0 \in [i] \cap \widetilde S} a_{i_0} \cdot Z + 2t\\
        & = d(\oldx_i) + n \cdot W + \sum_{i_0=1}^n i_0 \cdot X + \sum_{i_0=1}^i 2^{i_0} \cdot Y  + \sum_{i_0 \in [i] \cap \widetilde S} a_{i_0} \cdot Z + t\\
        & = d(\oldx_i) + \ell - \sum_{i_0 =i +1}^n 2^{i_0} \cdot Y + \bigl(\sum_{i_0 \in [i] \cap \widetilde S} a_{i_0} - t\bigr) \cdot Z
    \end{align*}
    which is at most $d(\oldx_i) + \ell$ for $i < n$ or $\sum_{i \in \widetilde S} a_{i_0} \le t$.

    For $\neg \oldx_{i}$, the calculations are identical to the one for $\oldx_i$, except that neither $\oldy_{i}$ nor $\neg \oldy_i$ is scheduled before~$\neg \oldx_i$ and the due date of $\neg \oldx_i$ is by $W$ smaller than the due date of $\oldx_i$.
    Consequently, $\neg \oldx_i$ is completed by time
    \begin{align*}
        \sum_{i_0=1}^n p(\mathcal F_{i_0}) + & \sum_{i_0 =1}^{i} \Bigl( p(\oldx_{i_0}) + p(\neg \oldx_{i_0}) \Bigr) + \sum_{i_0 = i+1}^n p(\neg \oldx_{i_0}) + \sum_{i_0 \in [i-1] \cap \widetilde S} p(\oldy_{i_0}) + \sum_{i_0 \in [i-1]\setminus \widetilde S} p(\neg \oldy_{i_0}) \\
        & = d(\neg \oldx_i) + n \cdot W + \sum_{i_0=1}^n i_0 \cdot X + \sum_{i_0=1}^{i-1} 2^{i_0} \cdot Y  + \sum_{i_0 \in [i-1] \cap \widetilde S} a_{i_0} \cdot Z + t\\
        & = d(\neg \oldx_i) + \ell - \sum_{i_0 =i}^n 2^{i_0} \cdot Y + \bigl(\sum_{i_0 \in [i-1] \cap \widetilde S} a_{i_0} - t\bigr) \cdot Z  \,.
    \end{align*}
    Thus, $\neg \oldx_{i}$ always has tardiness smaller than $\ell$.

    For the filler job~$F_i$, note that $F_i$ is completed before $\oldx_i$ for $\oldx_i \notin \widetilde{\mathcal J}$ but has the same due date.
    Thus, $F_i$ has tardiness smaller than $\ell$.

    We continue with job $\oldy_i \notin \widetilde{\mathcal{J}}$ or $\neg \oldy_i\notin \widetilde{\mathcal J}$.
    This job is completed after all filler jobs, jobs $\oldx_{i_0}$ and $\neg \oldx_{i_0}$ for $i_0\in[n]$, and $(n+i)$ jobs from $\{\oldy_{i_0}, \neg \oldy_{i_0}: i_0 \in [n]\}$.
    Consequently, the job is completed by time
    \begin{align*}
        n \cdot W + & 2\cdot \sum_{i_0 =1}^n i_0 \cdot X + \sum_{i_0=1}^n a_{i_0} + (n+i ) \cdot W + \sum_{i_0=1}^n 2^{i_0} \cdot Y+ \sum_{i_0=1}^i 2^{i_0} \cdot Y + \sum_{i_0=1}^n a_{i_0} \cdot Z \\
        & = D_1^* + (n+ i) \cdot W +  \sum_{i_0 =1}^n i_0 \cdot X + \sum_{i_0=1}^n 2^{i_0} \cdot Y+ \sum_{i_0=1}^i 2^{i_0} \cdot Y + 2t \cdot Z + t\\
        & = d(\neg \oldy_i) + n \cdot W + \sum_{i_0 = i}^n i_0 \cdot X +\sum_{i_0=1}^n 2^{i_0} \cdot Y + t \cdot Z\\
        & < d(\neg \oldy_i) + \ell < d(\oldy_i ) + \ell\,.
    \end{align*}
    Thus, jobs $\oldy_i$ and $\neg \oldy_i$ have tardiness smaller than $\ell$.
\end{proof}

We can now show the correctness of the reduction:
\begin{proof}[Proof of \Cref{thm:lex}]
    We start with the forward direction, so assume that there is a solution $S$ to the \textsc{Partition} instance.
    Let $\widetilde{\mathcal J} := \{\oldx_i, \oldy_i : i\in S\} \cup \{\neg \oldx_i, \neg \oldy_i : i \notin S\} \cup \mathcal F_1 \cup \ldots \cup \mathcal F_n$.
    By \Cref{lem:weak-early-y,lem:weak-early-x}, the schedule for~$\widetilde{\mathcal J}$ has $k$ tardy jobs.
    By \Cref{lem:weak-tardy}, the schedule has maximum tardiness at most $\ell$.

    We continue with the reverse direction.
    So assume that there is a schedule $\sigma$ with $k$ tardy jobs and maximum tardiness at most $\ell$.
    By \Cref{lem:structure-min-tardy-jobs,lem:y}, we may assume that $\sigma$ is the canonical schedule for some candidate set $\widetilde{\mathcal J}$.
    By \Cref{lem:structure-min-tardy-jobs}, we have $\sum_{i \in \widetilde S^*} a_i \ge t$.
    Further, we have $\sum_{i \in \widetilde S} a_i \le t$:
    If $\oldx_n \in \widetilde{\mathcal J}$, then this follows from \Cref{lem:weak-early-y}.
    Otherwise we have $ \oldx_n \notin \widetilde{\mathcal J}$ and the inequality follows from \Cref{lem:weak-tardy}.
    By \Cref{lem:weak-early-y}, we have $\widetilde S^* \subseteq \widetilde S$.
    Combining these statements, we get
    \[
        t \le \sum_{i \in \widetilde S^*} a_i \le \sum_{i \in \widetilde S} a_i \le t\,. 
    \]
    Consequently, we have $\widetilde S^* = \widetilde S$, and the inequalities hold with equality.
    Therefore, $\widetilde S^*$ is a solution to the \textsc{Partition} instance.
\end{proof}

\section{Conclusions}

In this paper we resolved one of the most fundamental problems in bicriteria scheduling which involves combining the maximal tardiness objective with the total number of tardy jobs objective. We proved that the lexicographic version of this problem (i.e.\ the $1||Lex(T_{\max},\sum U_j)$ problem) is strongly NP-complete when the maximal tardiness is used as the primary criterion, while it is at least weakly NP-complete when the number of tardy jobs is used as the primary criterion (i.e.\ the $1||Lex(\sum U_j, T_{\max})$ problem). We also classified the two other variants of the problem, $1|T_{\max} \le \ell, \sum U_j \le k|$ and $1||\alpha T_{\max}+\sum U_j$, as strongly NP-complete.

The first obvious question raised by our work is whether $1||Lex(\sum U_j, T_{\max})$ is strongly NP-complete or not. However, this is only one out of several other objectives that may be considered in future work. For example, there is lack of practical (polynomial time) approximation algorithms for solving bicriteria scheduling problems in general. It is interesting to see if one can design a polynomial time approximation scheme (PTAS) to solve our specific bicriteria problem, and if not maybe one can rule out the existence of such an algorithm.
Another very interesting research direction which is still unexplored is whether we can provide FPT algorithms to hard multicriteria scheduling problems with respect to some of the more natural parameters, e.g., the number of different due dates or the number of different processing times in the instance.

\section*{Acknowledgments}

We are very grateful to the anonymous reviewers of a previous version of the paper for their thoughtful and constructive feedback, pointing us to several mistakes in a previous draft.

\bibliography{bib}

\end{document}